\documentclass[acmsmall, nonacm]{acmart}
\AtBeginDocument{%
  }

\setcopyright{none}

\makeatletter
\AtEndPreamble{%
  \def\@acmdefinitionindent{\z@}%
}
\makeatother

\usepackage{booktabs}   
\usepackage{multirow}   
\usepackage{subcaption} 
\usepackage{listings}
\usepackage{parcolumns}
\usepackage[utf8]{inputenc}
\usepackage{mathpartir}

\usepackage[font=small,skip=0pt]{caption}
\usepackage{tikz-cd}
\usepackage{graphicx}

\usepackage{chor}
\usepackage{symbols}

\newcommand{\mech}{Mech\xspace}

\definecolor{kirby}{RGB}{215,72,148}
\definecolor{kirby-blue}{RGB}{68,102,191}
\definecolor{kirby-green}{RGB}{35,135,37}

\usepackage[normalem]{ulem}

\newcommand{\hlb}[1]{{\color{kirby-blue}#1}}

\usepackage{cleveref}
\crefname{example}{example}{examples}
\Crefname{example}{Example}{Examples}
\crefname{proposition}{proposition}{propositions}
\Crefname{proposition}{Proposition}{Propositions}

\usepackage{enumitem}

\begin{document}

\title{\mech: Mechanised Choreographic Programming}

\author{Xueying Qin}
\email{xyqin@imada.sdu.dk}
\orcid{0000-0003-4825-2023}
\affiliation{%
  \institution{University of Southern Denmark}
  \city{Odense}
  \country{Denmark}
}

\author{Marco Peressotti}
\email{peressotti@imada.sdu.dk}
\orcid{0000-0002-0243-0480}
\affiliation{%
  \institution{University of Southern Denmark}
  \city{Odense}
  \country{Denmark}
}

\author{Fabrizio Montesi}
\email{fmontesi@imada.sdu.dk}
\orcid{0000-0003-4666-901X}
\affiliation{%
  \institution{University of Southern Denmark}
  \city{Odense}
  \country{Denmark}
}


\begin{abstract}
\emph{Choreographic programming (CP)} is a programming paradigm for the correct-by-construction development of concurrent and distributed systems: programmers write the intended overall behaviour of a system from a global perspective in a \emph{choreography}, which is then automatically compiled into communicating endpoint programs by a procedure known as \emph{endpoint projection (EPP)}. 
The central promise is that the projected endpoint programs, when executed together, are behaviourally equivalent to the source choreography.

Fulfilling this promise becomes delicate for expressive CP languages. Existing mechanisations of CP treat only restricted fragments, while textbook and general purpose language implementations with rich features leave crucial interactions informal. 
In particular, general branching in knowledge of choice, general recursion, and nondeterministic choice in choreographies have not yet been integrated in a machine-checked theory.

We present \mech, a new mechanisation of CP in Lean 4 that captures these features. 
There are two central technical challenges in our development of \mech.
First, the sketched semantics from the literature does not correctly capture how nondeterministic choice interacts with concurrency.
We therefore formulate new semantics that align nondeterministic choreographic executions with the behaviours of projected endpoint programs.
Second, managing all these features in proofs is complex. 
We address this by uncovering new algebraic laws for choreographies, the operators used in their semantics, EPP, and their combinations.
Using our development, we prove completeness and soundness of EPP and derive communication safety and deadlock-freedom for projected networks, yielding the most extensive mechanised theory of CP to date.

\end{abstract}





\maketitle
\thispagestyle{plain}
\pagestyle{plain}

\section{Introduction}
\label{sec:introduction}

Concurrent and distributed systems are traditionally programmed from a \emph{local perspective}, where each endpoint is programmed independently and the overall system behaviour emerges from their interactions.
This approach can lead to two issues: 1) these systems are prone to errors like deadlocks and mismatched communications, which are hard to detect and debug~\cite{LLLG16}; 2) the emergent system behaviours are hard to reason about.

\emph{Choreographic programming} (CP) is a programming paradigm addressing these issues by taking a \emph{global perspective}~\citep{M13:phd,M23}.
A program in CP, called a \emph{choreography}, specifies the intended computation and, like a distributed protocol, the interactions among endpoints.
It can then be automatically compiled via \emph{endpoint projection (EPP)} into endpoint programs that follow the interactions prescribed by the choreography.
The correctness of EPP establishes \emph{behavioural equivalence} between choreographies and their projected programs: every behaviour of a choreography is reflected by its projected program, and vice versa. Consequently, the safety and liveness properties specified by a choreography are guaranteed in endpoint programs.
Many mainstream programming languages that implement CP, such as:
Clojure~\cite{LJ24},
Elixir~\cite{WG25},
Haskell~\cite{SKK23}, 
Java~\cite{GMP24},
Racket~\cite{BH25},
Rust and TypeScript~\cite{BKJSKN25}, and
Unison~\cite{C25}.

Despite this progress, the metatheory of CP remains prone to errors~\cite{M23Errata}, as observed for general concurrency theory~\cite{MS15,SY19}. Most importantly, existing formal developments capture only restricted fragments of choreographic languages. Essential expressive features such as \emph{general branching}, \emph{general recursion}, and \emph{nondeterministic choice} have never been formalised together in a machine-checked theory --- in fact, there have been no mechanised accounts of general branching or nondeterministic choice --- while these constructs substantially increase the expressivity of choreographies.
For instance, general branching enables multi-way protocol choreographies without cascades of binary selections, general recursion allows iterative choreographies to be expressed modularly and succinctly, and nondeterministic choice enables coordination patterns such as barriers and first-come-first-served interactions.

This absence turns out not to be accidental as formalising each of these features is challenging on its own. Moreover, their combination further compounds the complexity of the formalisation by introducing subtle interactions between control flow, concurrency, and EPP.
In particular, nondeterministic behaviours and multi-way branching introduce delicate dependencies during EPP to reconcile alternative behaviours.
When we attempted to mechanise expressive features in the Lean~4 theorem prover, we discovered that the semantics informally described in the literature~\cite{M23} \emph{is not deadlock free and does not correctly capture the semantics of compiled endpoint programs}.

In this paper we present \mech, a semantic foundation for CP, mechanised in Lean~4, revising the core theory in the textbook \emph{Introduction to Choreographies}~\cite{M23} with aforementioned expressive features.
Our paper makes the following contributions:
\begin{itemize}[nosep]
\item \textbf{Mechanised CP theory with expressive features.}
We present the first substantial mechanisation of CP, which contains over 40,000 lines of code in Lean~4, supporting general branching, general recursion with continuations, and nondeterministic choice.

\item \textbf{Correct semantics for nondeterministic choreographies.}
Mechanisation has revealed that the existing semantics suffers from deadlocks and does not capture compiled behaviours.
We introduce a new semantics and prove it correctly models nondeterministic behaviours.

\item \textbf{Machine-checked compilation correctness.}
We establish a strong behavioural equivalence between choreographies and their projected endpoint programs by proving the correctness of EPP.
We give formal characterisations of communication safety and deadlock-freedom and prove that they are guaranteed in
the endpoint programs generated by EPP.

\item \textbf{New metatheoretical results.}
Our work demonstrates that mechanisation can act not only as a tool for increasing confidence in programming language theory, but
also as a \emph{method for discovering new results}.
Our development reveals new algebraic laws governing choreographies and their compilation, providing modularity and compositionality to \mech.
\end{itemize}

The foundational literature of CP focuses on synchronous communication, as asynchrony is generally motivated by implementation-level concerns such as execution efficiency, rather than expressive power. We follow this scoping in \mech, and leave the discussion of asynchronous CP to \cref{sec:related-work,sec:conclusion:future-work}. All results in this paper are mechanised Lean 4, we only present main propositions and theorems in the following sections, omitting the details of the proofs and auxiliary lemmas.

\section{Overview}
\label{sec:overview}
Before diving into our formal development, we first give simple examples to informally illustrate three key components of choreographic programming: choreographies, EPP, and endpoint programs.
The following choreography, $\cn{rei}$, defines a simple protocol where two endpoints (we call them \emph{processes}) $\pid p$ and $\pid q$ collaborate to produce an invoice by $\pid p$ for a reimbursement at $\pid q$'s bank account.
\begin{example}[A Reimbursement Choreography]\label{ex:Cinv}
\[
\cn{rei} \defeq \ccom{\pid{q}}{\fname{coords()}}{\pid{p}}{c} \seq \ccom{\pid{p}}{\fname{invoice(c)}}{\pid{q}}{x} \seq \cnil 
\]
\end{example}
The first \emph{choreographic instruction}, $\ccom{\pid{q}}{\fname{coords()}}{\pid{p}}{c}$, is a communication:
the process $\pid q$ communicates the result of evaluating the local function $\fname{coords()}$ --- computing $\pid q$'s bank account coordinates for reimbursement --- to $\pid p$, which stores it in its local variable $\vname{c}$.
Next, $\pid{p}$ evaluates $\fname{invoice(c)}$ to generate an invoice for the reimbursement  at the received coordinates and sends it to $\pid{q}$, which stores it in the local variable $\vname{x}$. Finally, the choreography terminates with the \emph{terminated choreography} $\cnil$.

The EPP of the choreography above, written as $\epp{\cn{rei}}$, produces a collection of endpoint programs for $\pid{p}$ and $\pid{q}$. Such compositions of endpoint programs are called \emph{networks} \cite{M23}.
\begin{example}[EPP of the Reimbursement Choreography]\label{ex:EPPinv}
\[\epp{\cn{rei}} = \atomicn{p}{\precv{q}{c}\seq\psend{q}{\fname{invoice(c)}}\seq\pnil}\parcomp \atomicn{q}{\psend{p}{\fname{coords()}}\seq\precv{p}{\fname{x}}\seq\pnil} \]
\end{example}
This network consists of the local implementations for the processes $\pid p$ and $\pid q$ from the choreography, composed by the \emph{parallel composition operator} ($\!\parcomp\!$).
The notation $\atomicn{p}{P}$, where $\pid p$ is a process name and $P$ a process term, means that the process $\pid{p}$ executes its own program $P$, recalling distributed process calculi and ambients \cite{CG97,H07,M23}.

We will formally define the syntax and semantics for networks and process terms in \cref{sec:proc-net}, and EPP in \cref{sec:epp}. Here, we give an informal explanation of the example.
The first choreographic instruction $\ccom{\pid{q}}{\fname{coords()}}{\pid{p}}{c}$ is projected onto $\pid{q}$ as a \emph{process instruction} that gets its bank account coordinates by evaluating the local function $\fname{coords()}$ and sends the coordinates to $\pid{p}$, which is written as $\psend{p}{\fname{coords()}}$. Dually, its projection onto $\pid{p}$ is a process instruction that receives the coordinates from $\pid{q}$ and stores it in $\vname{c}$, written as $\precv{q}{c}$.
Similarly, the second choreographic instruction $\ccom{\pid{p}}{\fname{invoice(c)}}{\pid{q}}{x}$ is projected onto $\pid{p}$ as a process instruction sending the result of evaluating $\fname{invoice(c)}$ --- namely, the invoice generated with the received coordinates of $\pid{p}$ --- to $\pid{q}$. Its projection onto $\pid{q}$ is a process instruction receiving the invoice from $\pid{p}$ and storing it in $\vname{x}$.
Finally, the terminated choreography $\cnil$ is projected to the terminated process $\pnil$ on both processes.

As introduced in \cref{sec:introduction}, \emph{general branching} is an expressive feature supported by \mech, which allows processes to select among multiple labels instead of being restricted to binary choices. For instance, we can extend $\cn{rei}$ with \emph{conditional} and \emph{selection} instructions to describe a more complex protocol for invoicing, where $\pid{p}$ (a shop) decides  whether $\pid q$ (a customer) should perform a payment, receive a reimbursement, or neither (in which case no invoice is necessary).
\begin{example}[An Invoicing Choreography]\label{ex:CinvC}
\begin{spreadlines}{0pt}
\begin{align*}
\cn{inv} \defeq &\condif \pe{p}{\fname{pay()}} \condthen \csel{p}{q}{pay}\seq\ccom{\pid{p}}{\fname{invoice()}}{\pid{q}}{x} \seq \cnil\\
&\condelse \condif \pe{p}{\fname{reimburse()}} \condthen \csel{p}{q}{reimburse}\seq \ccom{\pid{q}}{\fname{coords()}}{\pid{p}}{c} \seq \ccom{\pid{p}}{\fname{invoice(c)}}{\pid{q}}{x} \seq \cnil\\
&\condelse \csel{p}{q}{skip}\seq\cnil
\end{align*}
\end{spreadlines}
\end{example}
The choreography $\cn{inv}$ has three cases, depending on the evaluations of the \emph{guard expressions} $\fname{pay()}$ and $\fname{reimburse()}$ by the \emph{guard process} $\pid p$.
Note that, depending on the local choice made by $\pid p$, $\pid q$ exhibits different behaviours: to receive an invoice from $\pid p$ (first branch), to send its bank coordinates to $\pid p$ and then receive an invoice from $\pid p$ (second branch), or to do nothing (third branch).
It is therefore essential that $\pid q$ gets to know what branch $\pid p$ has chosen, so that it can follow the choreography correctly.
The propagation of such \emph{knowledge of choice (KoC)} is a fundamental concept in choreographic languages and is typically implemented by means of \emph{selections} \cite{CDP12,M23}.
For example, in the first branch, the selection instruction $\csel{p}{q}{pay}$ means that $\pid p$ communicates the \emph{selection label} $\lbl{pay}$ to $\pid{q}$. Selection labels are essentially tags that the receiver can use to determine what it should subsequently execute.
In $\cn{inv}$, $\pid q$ receives a different label in each branch and can therefore unambiguously determine that $\pid p$ chooses the first branch if it receives $\lbl{pay}$, the second branch if it receives $\lbl{reimburse}$, or the third branch if it receives $\lbl{skip}$.
As shown in the example, selections are only needed when the guard process makes a choice that affects the behaviour of other processes. Note that KoC can be propagated transitively, for example, $\pid p$ could communicate the choice to another process $\pid r$, which could then inform $\pid q$.

The invoicing choreography $\cn{inv}$ is projected to the networks consisting of $\pid{p}$ and $\pid{q}$, shown below.
\begin{example}[EPP of the Invoicing Choreography]\label{ex:EPPinvC}
\begin{spreadlines}{0pt}
\begin{align*}
\left.
\begin{aligned}[c]
\epp{\cn{inv}} = \
\end{aligned}
\begin{aligned}[c]
\atomicn{p}{&\condif\fname{pay()}\condthen \\&\quad\pselg{q}{pay} \seq\psend{q}{\fname{invoice()}}\seq\pnil\\ &\condelse \condif \fname{reimburse()} \condthen \\&\quad\pselg{q}{reimburse}\seq\precv{q}{c}\seq\psend{q}{\fname{invoice(c)}}\seq\pnil \\ &\condelse \pselg{q}{skip}\seq\pnil }
\end{aligned}
\; \middle|\;
\begin{aligned}[c]
\atomicn{q}{\pid{p}\br\{&\lbl{pay}: \precv{p}{x}\seq\pnil\sep\\&\lbl{reimburse}: \psend{p}{\fname{coords()}} \seq \precv{p}{x} \seq \pnil \sep\\&\lbl{skip}: \pnil \}}
\end{aligned}
\right.
\end{align*}
\end{spreadlines}
\end{example}
The projection of the guard process $\pid p$ is simple and homomorphic, since it is the process evaluating guard expressions of the conditional.
If the guard expression $\fname{pay()}$ evaluates to $\vname{true}$, $\pid{p}$ sends the selection label $\lbl{pay}$ to $\pid{q}$, written as $\pselg{q}{pay}$, and then sends an invoice to $\pid{q}$, followed by a termination. Otherwise, if $\fname{reimburse()}$ is evaluated to $\vname{true}$, $\pid p$ sends the label $\lbl{reimburse}$ to $\pid q$, followed by the endpoint program projected from $\cn{rei}$ onto $\pid{p}$, explained in example~\ref{ex:EPPinv}. The third branch is taken if neither of the guard expressions are evaluated to $\vname{true}$, where $\pid p$ sends the label $\lbl{skip}$ to $\pid q$ and then terminates. The selections are crucial for $\pid p$ to give $\pid q$ the knowledge of which branch it has taken, so that $\pid q$ can behave accordingly.
The projection on $\pid q$ is more interesting: the three branches are \emph{merged} into a \emph{branching term} that awaits to receive a selection label from $\pid p$.
There are three selection labels of $\pid q$'s knowledge: $\lbl{pay}$, $\lbl{reimburse}$, and $\lbl{skip}$.
If $\pid q$ receives the label $\lbl{pay}$, then it expects to receive an invoice from $\pid p$ and stores it in the local variable $\vname{x}$, followed by a termination. If $\pid q$ receives $\lbl{reimburse}$, then it proceeds with the endpoint program projected from $\cn{rei}$ onto $\pid{q}$ as shown in example~\ref{ex:EPPinv}. If $\pid q$ receives $\lbl{skip}$, then it simply terminates. 
The symbols $\oplus$ and $\&$ are standard \cite{HYC16,M23}: they serve as mnemonic for sending an \emph{internal choice} ($\oplus$) and receiving an \emph{external choice} ($\&$), borrowing notation from logic \cite{G87}.
We will formally discuss the selection and merge language constructs under the context of EPP in \cref{sec:epp}.

As mentioned in \cref{sec:introduction}, our theory also supports \emph{choreographic choice}, which is projected to nondeterministic choice in endpoint programs, for modelling nondeterministic behaviours of systems. \emph{General recursion} is another important feature in our theory, which allows us to write recursive choreographies polymorphic over processes (with continuations). We will demonstrate these features as well as their interactions with other constructs in the following sections.

\subsection{Formal Preliminaries}
\paragraph{Parameters}
\label{sec:parameters}
Throughout our formal development, we consider processes running concurrently and interacting with each other by message passing. For generality, our models do not fix internal details of local computations at processes.
Specifically, we require the following parameters.
\begin{itemize}
\item A set of process names ($\setof{PName}$) used to identify processes, ranged over by $\pid p$.
\item A set of values ($\setof{Val}$) that processes can compute, ranged over by $v$.
\item A set of variables ($\setof{Var}$) that processes use to store values, ranged over by $x$.
\item A set of function names ($\setof{F}$) used to identify local functions that processes can invoke.
\item A local evaluation relation $\downarrow$ that models local computations at processes. This relation can be nondeterministic. We present it formally in \cref{chor:semantics}.
\item A set of selection labels ($\setof{SLabel}$) used by processes to offer alternative behaviours.
\item A set of procedure names ($\setof{ProcName}$) used to identify recursive procedures in programs.
\end{itemize}%

\paragraph{Function Support and Finite Functions}
\label{sec:finfun}
For any function $f$ whose codomain has a clear `zero' element ($0$), we write $\support(f)$ for its support: $ \support(f) \triangleq \{ a \in A \mid f(a) \neq 0 \}$.
Our development also makes use of functions with finite support. Lean's mathematical library (Mathlib) \cite{mathlib2020} provides a formalisation for such functions (\code{Finsupp}), but some of its accompanying operations are noncomputable on purpose to avoid decidability requirements (e.g., decidability of checking whether an element is equal to zero). Since we care about computing these functions and their transformations, we introduce a new formalisation called \code{FinFun} (for `finite functions'), which technically consists of a function $f : A \to B$ accompanied by a proof that only finitely many elements of $A$ are mapped to non-zero elements in $B$. This invariant is preserved by our APIs for finite functions. Both \code{FinFun} and these APIs have already been included in CSLib, the Lean Computer Science Library~\cite{cslib2026}. In the remainder, we annotate function names with a hat to indicate that they are finite functions, e.g., $\hat f$, and write their type as $\hat f : A \finfun B$.


\section{Choreographies}
\label{sec:chor}
We now move to our formal development, starting with the syntax and semantics of choreographies.

\subsection{Syntax}
\label{chor:syntax}
\begin{figure}
\vspace{-1.0em}
\begin{align*}
    \text{Choreographic Procedure Definitions ($\mathfrak{C}$)}\quad \mathscr{C} \ \metaDeff \ &\{X_i(\pids{p_i}) = C_i\}_{i \in \idx{I}}
\end{align*}
\vspace{-2.0em}
\begin{spreadlines}{0pt}
\begin{alignat*}{3}
    \mathrm{Choreography\;(\setof{Chor})}\ &&
    C \metaDeff&\ I \seq C \cmid \cnil \quad\quad\quad \mathrm{Expression\;(\setof{Expr})}\
    e \metaDeff \ v \cmid x \cmid f(\vec{e})\\
    \mathrm{Choreographic\;Instruction\;(\setof{Instr})}\ &&
    I \metaDeff&\ \pid{p}.x \metaDef e \cmid \gencom \cmid \gensel \\ && \cmid &\ \condif\pid{p}.e\condthen C_1\condelse C_2  \cmid C_1 \choice_\pid{p} C_2 \cmid X(\pids{p}) \cmid \pid{q} : X(\pids{p}).C
\end{alignat*}
\end{spreadlines}
\caption{\label{chor:syntax:fig}Syntax of Choreographies in \mech}
\vspace{-1.0em}
\end{figure}
\Cref{chor:syntax:fig} shows the formal syntax of choreographies in \mech.
Choreographies ($C$) and their instructions ($I$) are mutually defined.
A choreography is either a sequence of instructions, constructed with the typical sequence operator ($;$), or a $\cnil$, representing a terminated choreography.
For readability, we often omit trailing $\cnil$s in examples.
There are three basic instructions.
In a \emph{local assignment} $\assign{p}{x}{e}$, a process $\pid p$ evaluates the expression $e$ and stores the result in its local variable $x$. In a (value) \emph{communication} $\gencom$, $\pid p$ sends the result of evaluating $e$ to $\pid q$, which stores it in a local variable $x$. In a \emph{selection} $\gensel$, $\pid p$ sends a selection label $\albl$ to $\pid q$ (recall selection labels in \cref{sec:overview}).

There are two instructions expressing choices in \mech. One is the standard \emph{conditional} construct for choreographies, written $\condif\pid{p}.e\condthen C_1\condelse C_2$, where a single process $\pid p$ evaluates $e$ to decide whether the choreography should continue as $C_1$ (if $e$ is evaluated to true) or $C_2$ (otherwise), as shown in the invoicing choreography $\cn{inv}$ in example~\ref{ex:CinvC}.
The other is the \emph{choreographic choice} $C_1 \choice_\pid{p} C_2$ where -- differently from conditionals -- the resolution of whether $\pid{p}$ continues with $C_1$ or $C_2$ is determined by the behaviour of a system instead of a single process.
As such, choreographic choice can express nondeterministic system behaviour that cannot be modelled by conditionals (even considering that guard expressions are evaluated nondeterministically), including patterns like barriers, locks, and first-come-first-served~\cite{M23}. We show this in the next example.
\begin{example}[A First-Come-First-Served Choreography~\cite{M23}]\label[example]{ex:chor-fcfs}
The choreography $\cn{fcfs}$ models two processes $\pid{p_1}$ and $\pid{p_2}$ competing to gain access to shared resource on another process $\pid{s}$, by sending requests to $\pid{s}$ indicated by a selection label $\lbl{req}$.
$\pid s$ replies $\lbl{ok}$ to the requests in a first-come, first-served manner.
The winner can work with $\pid{s}$, modelled by the \emph{procedure} $W$.
\begin{spreadlines}{0pt}
\begin{align*}
\cn{fcfs} &\defeq (\csel{p_1}{s}{req}\seq \csel{s}{p_1}{ok}\seq W(\pid{p_1},\pid{s})\seq\csel{p_2}{s}{req}\seq \csel{s}{p_2}{ko}\\
&\choicep{s}\csel{p_2}{s}{req}\seq \csel{s}{p_2}{ok}\seq W(\pid{p_2},\pid{s})\seq\csel{p_1}{s}{req}\seq \csel{s}{p_1}{ko})
\end{align*}
\end{spreadlines}
\end{example}

The next example applies conditionals and choreographic choice to model a triage process in the healthcare domain \cite{PZWW22}.
\begin{example}[A Triage Choreography]\label[example]{ex:chor-triage}
The choreography $\cn{triage}$ models a typical triage scenario where a doctor ($\pid d$) awaits the results of a chest X-ray (by $\pid{rx}$) and a blood test (by $\pid {bw}$). If the X-ray results arrive first, the doctor must react immediately (first branch) to determine (conditional) whether they call for an urgent surgery (decision communicated to $\pid {op}$) \cite{ATLS2018,StatPearls2023}.
Otherwise, if the doctor receives the blood test results first, the decision is postponed until the X-ray results are available.
\begin{spreadlines}{0pt}
\begin{align*}
    \cn{triage}&\defeq \ccom{rx}{\fname{test()}}{d}{x}\seq (\condif \pe{d}{\fname{bad(x)}} \condthen \csel{d}{op}{on} \condelse \csel{d}{op}{no}) \seq \ccom{bw}{\fname{test()}}{d}{y}\\
    &\choicep{d} \ccom{bw}{\fname{test()}}{d}{y} \seq \ccom{rx}{\fname{test()}}{d}{x}\seq (\condif \pe{d}{\fname{bad(x)}} \condthen \csel{d}{op}{on}\condelse \csel{d}{op}{no})
\end{align*}
\end{spreadlines}
\end{example}

Lastly, we have \emph{procedure calls} and \emph{runtime terms} together with a set of \emph{choreographic procedure definitions}
for modelling \emph{general recursion}.
In a \emph{procedure call} $X(\pids{p})$, the processes $\pids p$ execute the procedure $X$. In \cref{ex:chor-fcfs}, $W(\pid{p_1},\pid{s})$ is a procedure call where the winning process $\pid{p_1}$ can access the shared resource on $\pid{s}$.
We define procedures as equations of the form $X(\pids{p}) = C$, where $\pids p$ is a sequence of process parameters of $X$ (procedures are polymorphic over processes) and $C$ is the body of $X$. A set of choreographic procedure definitions $\mathscr{C}$ collects \emph{distinct} procedures defined as equations.
We implement $\mathscr C$ as a list in Lean 4 to afford straightforward inductive reasoning.

\begin{example}[A Recursive First-Come-First-Served Choreographic Procedure Definition]\label[example]{ex:chor-fcfs-rec}
Using general recursion, we can define a recursive procedure that extends the shared resource choreography from \cref{ex:chor-fcfs} to repeat the competition for the acquisition of the resource ($\pid s$ acts as a reusable lock).
\begin{spreadlines}{0pt}
\begin{align*}
X(\pid{s},\pid{p_1},\pid{p_2}) &= (\csel{p_1}{s}{req}\seq \csel{s}{p_1}{ok}\seq W(\pid{p_1},\pid{s})\seq\csel{p_2}{s}{req}\seq \csel{s}{p_2}{ko}\\
&\choicep{s}\csel{p_2}{s}{req}\seq \csel{s}{p_2}{ok}\seq W(\pid{p_2},\pid{s})\seq\csel{p_1}{s}{req}\seq \csel{s}{p_1}{ko})\seq X(\pid{s},\pid{p_1},\pid{p_2})
\end{align*}
\end{spreadlines}
\end{example}

The last term -- $\pid{q} : X(\pids{p}).C$ -- is a runtime term that is not intended for programmers. It is a standard technical construct used in our semantics to address the textbook problem of modelling that processes can enter a procedure concurrently, without the need for any coordination~\cite{M23}. Specifically, it indicates that one or more processes have already started procedure $X$ while $\pid q$ has yet to do so. The $C$ in a runtime term is irrelevant for the semantics of choreographies, but serves as a useful annotation for EPP, which we discuss in more detail in \cref{sec:epp}.
\begin{figure}[b]
\vspace{-1.0em}
\begin{align*}
\begin{aligned}[c]
\pn(\cnil) &\triangleq \emptyset
\end{aligned}
\quad
\begin{aligned}[c]
\pn(I \seq C) &\triangleq \pn(I) \cup \pn(C)
\end{aligned}
\quad
\begin{aligned}[c]
\pn(\pid{p}.x \metaDef e) &\triangleq \{\pid{p}\}
\end{aligned}
\end{align*}
\vspace{-1.7em}
\begin{align*}
\begin{aligned}[c]
\pn(\gencom) &\triangleq \{\pid{p}, \pid{q}\}
\end{aligned}
\;\
\begin{aligned}[c]
\pn(\gensel) &\triangleq \{\pid{p}, \pid{q}\}
\end{aligned}
\;\
\begin{aligned}[c]
\pn(X(\pids{p})) &\triangleq \{\pids{p}\}
\end{aligned}
\;\
\begin{aligned}[c]
\pn(\pid{q} : X(\pids{p}).C ) &\triangleq \{\pid{q}\}
\end{aligned}
\end{align*}
\vspace{-1.7em}
\begin{align*}
\begin{aligned}[c]
\pn(\condif \pid{p}.e \condthen C_1 \condelse C_2) &\triangleq \{\pid{p}\} \cup \pn(C_1) \cup \pn(C_2)
\end{aligned}
\;\;
\begin{aligned}[c]
\pn(C_1 \choice_\pid{p} C_2) &\triangleq \{ \pid{p} \} \cup \pn(C_1) \cup \pn(C_2)
\end{aligned}
\end{align*}
\caption{Process Names in Choreographies and Instructions}
\label{chor:syntax:pn}
\vspace{-1.0em}
\end{figure}

We also define functions $\pn(C)$ and $\pn(I)$ for gathering the process names in choreographies and instructions
(\cref{chor:syntax:pn}), which are used by our semantics.

\subsection{Semantics}
\label{chor:semantics}
We define the operational semantics of choreographies as a \emph{labelled transition system (LTS)} of which states are configurations of the form $\langle C, \cstore, \mathscr{C} \rangle$, where $C$ is a choreography, $\cstore$ is a global store mapping process names to their \emph{local stores}, and $\mathscr{C}$ is a set of \emph{choreographic procedure definitions}.
A transition takes the form of $\langle C, \cstore, \mathscr{C}\rangle \xrightarrow[]{\mu} \langle C', \cstore', \mathscr{C}\rangle$, where $\mu$ is a transition label. Before presenting the rules of our semantics, we first introduce some important notations and definitions.

\subsubsection{Stores and Evaluation of Expressions}
\label{chor:semantics:stores}
In a choreographic configuration $\langle C, \cstore, \mathscr{C} \rangle$,
the global store $\cstore$ is a function mapping process names to \emph{local stores}, representing the memory state of the system. A local store $\sigma$ on a process is a function mapping variables to values, representing the local memory of a process. Formally they are defined as:
\begin{align*}
\begin{aligned}[c]
    \mathrm{Global\, Store\, (\setof{GStore})}\quad \cstore : \setof{PName} \to \setof{LStore}
\end{aligned}
\quad\quad\quad\quad
\begin{aligned}[c]
    \mathrm{Local\, Store\, (\setof{LStore})} \quad\sigma : \setof{Var} \to \setof{Val}
\end{aligned}
\end{align*}
\\[-3.5ex]
There are two important operations on stores, which are the \emph{update} of a store's content and the \emph{evaluation} of an expression under a given store.
For the update, we write $\sigma[x \mapsto v]$ to denote the update of a local store $\sigma$ with variable $x$ mapped to value $v$, and $\cstore[\pid{p}.x \mapsto v]$ for updating the global store $\cstore$ such that the local store of $\pid{p}$ is updated with variable $x$ mapped to value $v$. (For detailed definitions, see definition~\ref{app:def:store-update}).
For the evaluation, we write $\sigma \vdash e \downarrow v$ to denote that expression $e$ evaluates to value $v$ under store $\sigma$, and the only assumption we make is that the evaluation relation $\downarrow$ is total. Thus in particular, such an evaluation can be nondeterministic as introduced in \cref{chor:syntax}.

\subsubsection{Transition Labels}
\label{chor:semantics:tl}
\begin{figure}
\vspace{-1.0em}
\begin{align*}
\begin{aligned}[c]
\text{Choice Label ($\mathscr{L}$)} \quad \mathscr{l} \ &\metaDeff \ \mathit{left} \cmid \mathit{right}
\end{aligned}
\quad\quad\quad
\begin{aligned}[c]
\text{Choice Map} \quad \zeta : \textbf{PName} \finfun [\mathscr{L}]
\end{aligned}
\end{align*}
\vspace{-2.0em}
\begin{spreadlines}{0pt}
\begin{alignat*}{3}
\text{Pre-Transition Label} \quad && \rho \ \metaDeff& \ \tlassign{p} \cmid \pid{p}.v \tto \pid{q} \cmid \gensel \cmid \tlthen{p} \cmid \tlelse{p} \\ && \cmid &\ \tlcall{p} \cmid \tlskip{p}\\
\text{Transition Label} \quad && \mu \ \metaDeff& \ (\rho, \zeta)
\end{alignat*}
\end{spreadlines}
\caption{Transition Labels}
\label{chor:semantics:transition-labels}
\vspace{-1.0em}
\end{figure}
We define the transition labels for the LTS in \cref{chor:semantics:transition-labels}.
A key novelty of our definition is that transition labels carry \emph{choice maps} that record the path of resolving choreographic choices of each involved process. This is crucial for defining a correct semantics that accounts for nested choices, as we will discuss in example~\ref{ex:chor-nested-choice}. Note that the semantics of choreographies and networks use the same transition labels, which is important for establishing the behavioural equivalence between choreographies and their projected networks in \cref{sec:epp}.

A transition label $\mu$ is a pair $(\rho, \zeta)$ where $\rho$ is a \emph{pre-transition label} and $\zeta$ is a choice map.
A pre-transition label expresses an action performed by process(es) in a transition.
For instance, $\tlassign{p}$ denotes that $\pid p$ performs a local assignment and $\pid{p}.v \tto \pid{q}$ models a communication of value $v$ from $\pid{p}$ to $\pid{q}$.
A \emph{choice map} ($\zeta$) is a \emph{finite function} introduced in \cref{sec:finfun}. It maps process names to lists of \emph{choice labels}, which can be either $\mathit{left}$ or $\mathit{right}$ (indicating which branch of the choice has been chosen). We write $[\mathscr{L}]$ for the set of such lists, where the zero element is the empty list. We also define a function adding a choice label into a transition label as shown in \cref{chor:semantics:add-choice}.
\begin{align}
(\rho, \zeta).\mathit{addChoice}\; \pid{p}\; \mathscr{l} \triangleq (\rho, \zeta[\pid{p}\mapsto \mathscr{l} \cons (\zeta\; \pid{p})])
\label{chor:semantics:add-choice}
\end{align}
\Cref{chor:semantics:pntl} shows the function $\pn$ gathering process names in pre-transition and transition labels.
\begin{figure}[b]
\vspace{-1.0em}
\begin{align*}
\begin{aligned}[c]
\pn((\rho, \zeta)) &\triangleq \pn(\rho) \cup \textbf{dom}(\zeta)\\
\pn(\tlskip{p}) &\triangleq \{ \pid{p} \}
\end{aligned}
\;
\begin{aligned}[c]
\pn(\tlassign{p}) &\triangleq \{ \pid{p} \}\\
\pn(\tlcall{p}) &\triangleq \{ \pid{p} \}
\end{aligned}
\;
\begin{aligned}[c]
\pn(\pid{p}.v \tto \pid{q}) &\triangleq \{\pid{p} , \pid{q}\}\\
\pn(\gensel) &\triangleq \{\pid{p} , \pid{q}\}
\end{aligned}
\;
\begin{aligned}[c]
\pn(\tlthen{p}) &\triangleq \{ \pid{p} \}\\
\pn(\tlelse{p}) &\triangleq \{ \pid{p} \}
\end{aligned}
\end{align*}
\caption{Process Names in Transition Labels}
\label{chor:semantics:pntl}
\vspace{-1.0em}
\end{figure}

\subsubsection{Process Name Substitution}
\label{chor:substitutions}
To model parametric procedure calls for recursions, we define the name substitution for process names in choreographies and instructions,
written as $C[\pid{r}/\pid{s}]$ and $I[\pid{r}/\pid{s}]$, meaning that we replace all occurrences of $\pid{s}$ with $\pid{r}$ in $C$ and $I$ respectively.
We write $C[\pids{r}/\pids{s}]$ as a shorthand of the name substitution in a choreography for many process names all at once, where the length of $\pids{r}$ and $\pids{s}$ must be the same.
Please refer to definitions~\ref{chor:semantics:pid-pids-substitution} and \ref{chor:semantics:chor-substitution} for the details.

\subsubsection{Sequential Composition}
In \cref{chor:semantics:seqcomp}, we define the \emph{sequential composition operator} ($\fatsemi$) to handle the continuation of choreographies in transitions involving conditionals, choices, and/or procedure calls.
Note that $\fatsemi$ is a not part of the syntax of choreographies, rather, it is a semantic operator.
Hence, $C_1 \fatsemi C_2$ is not a choreography term but a result of applying $\fatsemi$ to $C_1$ and $C_2$.
\begin{figure}[t]
\vspace{-1.0em}
\begin{align*}
\begin{aligned}[c]
    \cnil \fatsemi C &\triangleq C\\[-0.5ex]
    (I \seq C') \fatsemi C &\triangleq (I \fatsemi C) \seq (C' \fatsemi C)
\end{aligned}
\quad\quad\quad\quad
\begin{aligned}[c]
    I \fatsemi C &\triangleq 
    \begin{cases}
    \ \pid{q} : X(\pids{p}).(C' \fatsemi C) &\text{if $I = \pid{q} : X(\pids{p}).C' $}\\[-0.5ex]
    \ I &\text{otherwise}
    \end{cases}
\end{aligned}
\end{align*}
\caption{Sequential Composition of Choreographies}
\label{chor:semantics:seqcomp}
\vspace{-1.5em}
\end{figure}
The set of choreographies together with the sequential composition operator $\fatsemi$ forms a \emph{monoid}, where $\cnil$ is the \emph{identity element} and $\fatsemi$ is the \emph{associative operator}, of which the details are shown in proposition~\ref{app:chor:seqcomp:prop}.

\begin{figure}[b]
\vspace{-1.5em}
\begin{mathpar}
    \inferrule*[right={\scriptsize assign}] { \cstore(\pid{p}) \vdash e \downarrow v }
    {\langle \pid{p}.x \metaDef e \seq C, \cstore, \mathscr{C} \rangle \xrightarrow[]{\tlassign{\pid{p}}} \langle C, \cstore[\pid{p}.x \mapsto v], \mathscr{C} \rangle}

    \inferrule*[right={\scriptsize com}] { \cstore(\pid{p}) \vdash e \downarrow v } 
    {\langle \gencom \seq C, \cstore, \mathscr{C} \rangle \xrightarrow[]{\pid{p}.v \tto \pid{q}} \langle C, \cstore[\pid{q}.x \mapsto v], \mathscr{C} \rangle}
    \quad\inferrule*[right={\scriptsize sel}] {  } 
    {\langle \gensel \seq C, \cstore, \mathscr{C} \rangle \xrightarrow[]{\gensel} \langle C, \cstore, \mathscr{C} \rangle}

    \inferrule*[right={\scriptsize cond-then}] { \cstore(\pid{p}) \vdash e \downarrow \mathit{true}} 
    {\langle \condif \pid{p}.e\condthen C_1\condelse C_2 \seq C, \cstore, \mathscr{C} \rangle \xrightarrow[]{\tlthen{p}} \langle C_1 \fatsemi C, \cstore, \mathscr{C} \rangle }

    \inferrule*[right={\scriptsize cond-else}] { \cstore(\pid{p}) \vdash e \downarrow v \\ v \neq \mathit{true}} 
    {\langle \condif \pid{p}.e\condthen C_1\condelse C_2 \seq C, \cstore, \mathscr{C} \rangle \xrightarrow[]{\tlelse{p}} \langle C_2 \fatsemi C, \cstore, \mathscr{C} \rangle }
    
    \inferrule*[right={\scriptsize choice-l}]{\langle C_1, \cstore, \mathscr{C}\rangle \xrightarrow[]{\mu} \langle C_1', \cstore', \mathscr{C}\rangle \\ \pid{p}\in \pn(\mu)}{\langle C_1 \choice_\pid{p} C_2\seq C,  \cstore, \mathscr{C}\rangle \xrightarrow[]{\addChoice{p}{left}} \langle C_1' \fatsemi C, \cstore', \mathscr{C} \rangle}
    
    \inferrule*[right={\scriptsize choice-r}]{\langle C_2, \cstore, \mathscr{C}\rangle \xrightarrow[]{\mu} \langle C_2', \cstore', \mathscr{C}\rangle \\ \pid{p}\in \pn(\mu)} { \langle C_1 \choice_\pid{p} C_2\seq C, \cstore, \mathscr{C}\rangle \xrightarrow[]{\addChoice{p}{right}} \langle C_2' \fatsemi C, \cstore', \mathscr{C} \rangle }

    \inferrule*[right={\scriptsize skip-l}]{\pid{p} \notin \pn(C_1) \cup \pn(C_2)}{ \langle C_1 \choice_\pid{p} C_2\seq C, \cstore, \mathscr{C}\rangle \xrightarrow[]{\tlskip{p}} \langle C_1 \fatsemi C, \cstore, \mathscr{C} \rangle }
    \;\;
    \inferrule*[right={\scriptsize skip-r}]{\pid{p} \notin \pn(C_1) \cup \pn(C_2)}{ \langle C_1 \choice_\pid{p} C_2\seq C, \cstore, \mathscr{C}\rangle \xrightarrow[]{\tlskip{p}} \langle C_2 \fatsemi C, \cstore, \mathscr{C} \rangle }

    \inferrule*[right={\scriptsize call-first}] { X(\pids{q}) = C \in \mathscr{C} \\ \pids{p} = \pid{p}_1, \ldots, \pid{p}_n \\ i \in [1,n] }
    {\langle X(\pids{p});C', \cstore, \mathscr{C} \rangle \xrightarrow[]{\tlcall{\pid{p}_i}} \\\langle \pid{p}_1 : X(\pids{p}).C'; \ldots; \pid{p}_{i-1} : X(\pids{p}).C'; \pid{p}_{i+1} : X(\pids{p}).C'; \ldots; \pid{p}_n : X(\pids{p}).C'; (C[\pids{p}/\pids{q}] \fatsemi C'), \cstore, \mathscr{C} \rangle}

    \inferrule*[right={\scriptsize call-enter}] {  }
    {\langle \pid{q} : X(\pids{p}).C'; C, \cstore, \mathscr{C} \rangle \xrightarrow[]{\tlcall{\pid{q}}} \langle C, \cstore, \mathscr{C} \rangle}
    \;\;\inferrule*[right={\scriptsize delay}] {\langle C, \cstore, \mathscr{C}\rangle \xrightarrow[]{\mu} \langle C', \cstore', \mathscr{C}\rangle \\ \pn(I)\disj\pn(\mu)} 
    {\langle I\seq C, \cstore, \mathscr{C}\rangle \xrightarrow[]{\mu} \langle I\seq C', \cstore', \mathscr{C}\rangle}

    \inferrule*[right={\scriptsize delay-cond}] {\langle C_1, \cstore, \mathscr{C}\rangle \xrightarrow[]{\mu} \langle C_1', \cstore', \mathscr{C}\rangle \\ \langle C_2, \cstore, \mathscr{C}\rangle \xrightarrow[]{\mu} \langle C_2', \cstore', \mathscr{C}\rangle \\ \pid{p}\notin \pn(\mu)}
    {\langle \condif \pid{p}.e\condthen C_1\condelse C_2 \seq C, \cstore, \mathscr{C} \rangle \xrightarrow[]{\mu} \langle \condif \pid{p}.e\condthen C_1'\condelse C_2' \seq C, \cstore', \mathscr{C} \rangle}
    
    \inferrule*[right={\scriptsize delay-choice}]{\langle C_1, \cstore, \mathscr{C}\rangle \xrightarrow[]{\mu} \langle C_1', \cstore', \mathscr{C}\rangle \\ \langle C_2, \cstore, \mathscr{C}\rangle \xrightarrow[]{\mu} \langle C_2', \cstore', \mathscr{C}\rangle \\ \pid{p} \notin \pn(\mu)}{ \langle C_1 \choice_\pid{p} C_2\seq C, \cstore, \mathscr{C}\rangle \xrightarrow[]{\mu} \langle C_1' \choice_\pid{p} C_2'\seq C, \cstore', \mathscr{C} \rangle }
\end{mathpar}
\caption{Semantics of Choreographies}
\label{chor:semantics:lts:fig}
\vspace{-1.0em}
\end{figure}

\subsubsection{LTS for Choreographies}
\label{chor:semantics:lts}
We give the formal semantics of choreographies as a LTS over configurations in \cref{chor:semantics:lts:fig}.
For better readability, we write only the pre-transition label $\rho$ in the rules when the choice map $\zeta$ is empty. For instance, we write $\tlassign{p}$ as a short-hand of $(\tlassign{p}, \emptyset)$.

We first introduce three basic rules.
Rule \textsc{assign} models the execution of a local assignment. 
It states that if an expression $e$ is evaluated to a value $v$ in the local store of $\pid{p}$,
then the choreography can perform a transition with label $\tlassign{\pid{p}}$ to update the store of $\pid{p}$ by mapping $x$ to $v$.
Rule \textsc{com} models the execution of a communication between two processes.
Different from \textsc{assign}, instead of updating the local store of $\pid{p}$, the choreography performs a transition with label $\pid{p}.v \tto \pid{q}$ to communicate the result of evaluating $e$, from $\pid{p}$ to $\pid{q}$ and update the store of $\pid{q}$.
Rule \textsc{sel} is similar to \textsc{com} but for communicating selection labels between two processes, stating that a choreography can perform a transition with label $\gensel$ to send a label $\albl$ from $\pid{p}$ to $\pid{q}$.

Rules \textsc{cond-then} and \textsc{cond-else} for modelling the execution of a conditional.
The sequential composition $(\fatsemi)$ is used to handle the continuation of the choreography after executing a conditional.
Specifically, rule \textsc{cond-then} states that if an guard expression $e$ is evaluated to $\vname{true}$ in the store of a guard process $\pid{p}$, then the choreography can perform a transition with label $\tlthen{p}$ to continue with then-branch $C_1$ followed by the continuation $C$. If the expression $e$ is not evaluated $\vname{true}$, according to rule \textsc{cond-else}, the choreography will instead perform a transition with label $\tlelse{p}$ to continue with else-branch $C_2$ followed by the continuation $C$.

We have four rules for modelling the execution of a choreographic choice.
Rule \textsc{choice-l} states that if the left branch $C_1$ can perform a transition involving the choice-resolving process $\pid p$, then the choice is resolved to the left branch with a transition labelled $\addChoice{p}{left}$ and continues as $C_1' \fatsemi C$. The transition label $\addChoice{p}{left}$ records that the choice is resolved to the left branch for $\pid{p}$.
Rule \textsc{choice-r} is symmetric to \textsc{choice-l} but resolves the choice to the right branch of the choice.
Recall the first-come-first-served choreography $\cn{fcfs}$ in example~\ref{ex:chor-fcfs}, suppose that $\pid{p_1}$'s request reaches $\pid r$ first, according to the rule \textsc{choice-l} the choreography can perform a transition with the transition label $\addChoiceg{\csel{p_1}{s}{req}}{s}{left}$ to continue with the left branch of the choice:
\begin{spreadlines}{0pt}
\begin{align*}
    (&\csel{p_1}{s}{req}\seq \csel{s}{p_1}{ok}\seq W(\pid{p_1},\pid{s})\seq\csel{p_2}{s}{req}\seq \csel{s}{p_2}{ko}\\
\choicep{s} \ &\csel{p_2}{s}{req}\seq \csel{s}{p_2}{ok}\seq W(\pid{p_2},\pid{s})\seq\csel{p_1}{s}{req}\seq \csel{s}{p_1}{ko})\\
\xrightarrow[]{\addChoiceg{\csel{p_1}{s}{req}}{s}{left}} \ &\csel{s}{p_1}{ok}\seq W(\pid{p_1},\pid{s})\seq\csel{p_2}{s}{req}\seq \csel{s}{p_2}{ko}
\end{align*}
\end{spreadlines}
Representing the choice map as a finite function is crucial for nested choices, as it records choices resolved at each participating process w/o imposing an ordering on the processes (cf. example~\ref{ex:chor-nested-choice}).
\begin{example}[Nested Choreographic Choices]\label{ex:chor-nested-choice}
Although $\cn{pq}$ and $\cn{qp}$ differ syntactically in their nesting structure, they exhibit identical behaviour: their transitions are independent of the nesting of choices, as all choices are resolved jointly by $\pid p$ and $\pid q$. In fact, they are projected to the same network. If the choice is resolved to the left branch for $\pid{p}$ and the right branch for $\pid{q}$, then both choreographies can perform a transition to $\ccom{p}{1}{q}{y}$ with the label $\addChoiceg{(\addChoiceg{\tlcom{p}{1}{q}}{p}{left})}{q}{right}$.
\begin{spreadlines}{0pt}
\begin{align*}
\cn{pq} &\defeq (\ccom{p}{1}{q}{x}\choicep{q}\ccom{p}{1}{q}{y}) \choicep{p} (\ccom{p}{2}{q}{x}\choicep{q}\ccom{p}{2}{q}{y})\\
\cn{qp} &\defeq (\ccom{p}{1}{q}{x}\choicep{p}\ccom{p}{2}{q}{x}) \choicep{q} (\ccom{p}{1}{q}{y}\choicep{p}\ccom{p}{2}{q}{y})      
\end{align*}
\end{spreadlines}
\end{example}

Garbage-collection rules \textsc{skip-l} and \textsc{skip-r} play a crucial role in ensuring deadlock freedom. They model the case where the choice-resolving process $\pid{p}$ is not involved in the next transition of either branch. The choreography therefore performs a transition labelled $\tlskip{p}$ and proceeds with a chosen branch followed by the continuation C, demonstrated in example~\ref{ex:chor-gc}.
\begin{example}[A Choreographic Choice Requires Garbage Collection Rules]\label{ex:chor-gc}
The choreography
$\cn{gc} \defeq (\assign{p}{x}{42} \choicep{q} \assign{p}{x}{42}) \seq \ccom{\pid{p}}{7}{\pid{q}}{y}$ contains a choice at a process $\pid q$ that is not involved in either branch.
Without garbage collection rules, $\cn{gc}$ is stuck. However, with rule \textsc{skip-l} or \textsc{skip-r}, it can 
perform a transition with label $\tlskip{q}$ to continue with either the left or the right branch.
\[ (\assign{p}{x}{42} \choicep{q} \assign{p}{x}{42}) \seq \ccom{\pid{p}}{7}{\pid{q}}{y} \xrightarrow[]{\tlskip{q}} \assign{p}{x}{42} \fatsemi \ccom{\pid{p}}{7}{\pid{q}}{y}\]
\end{example}

The execution of recursive procedure calls is modelled by two rules, \textsc{call-first} and \textsc{call-enter}. Rule \textsc{call-first} allows any process $\pid{p}_i$ in the arguments of a procedure call $X(\pids{p})$ to initiate the call by performing a transition labelled $\tlcall{\pid{p}_i}$. The remaining processes $\pid{p}_j$ ($j \neq i$) become runtime terms $\pid{p}_j : X(\pids{p}).C'$, indicating that they have yet to enter the procedure. The continuation $C'$ is preserved, since from the perspective of these processes the procedure call has not yet begun. Once every process in $\pids{p}$ has entered the procedure, execution continues as $C[\pids{p}/\pids{q}] \fatsemi C'$, where the actual arguments $\pids{p}$ are substituted for the formal parameters $\pids{q}$.
Rule \textsc{call-enter} removes a runtime term $\pid{q} : X(\pids{p}).C'$ and continues with $C$.
This two-stage treatment of recursive procedure calls via runtime terms is essential for proving the correctness of EPP and establishing the behavioural equivalence between choreographies and their projections, as shown in \cref{sec:epp}.

Finally, our semantics models that processes execute concurrently by allowing  instructions independent on processes to be executed out of order.
For instance, given a choreography:
$\ccom {p} {e_1} {q} {x} \seq \pid{r}.y \metaDef e_2 \seq \cnil$,
where $\pid{r}$ is different from both $\pid{p}$ and $\pid{q}$, i.e., the communication instruction $\ccom {p} {e_1} {q} {x}$ and the assignment instruction $\pid{r}.y \metaDef e_2$ are independent.
Thus, they may be executed in either order. Regardless of the execution order, the choreography evaluates to the same result.
Formally, the out-of-order execution is captured by three rules: \textsc{delay}, \textsc{delay-cond}, and \textsc{delay-choice}.
Rule \textsc{delay} states that if $C$ can perform a transition with a label $\mu$ to $C'$, and none of the processes occurring in $I$ is mentioned in $\mu$, then $I \seq C$ can perform the same transition to $I \seq C'$.
Rule \textsc{delay-cond} states that if both branches of a conditional $C_1$ and $C_2$ can perform a transition with the same transition label $\mu$ to $C_1'$ and $C_2'$ respectively, and the guard process $\pid{p}$ is not mentioned in $\mu$, then the conditional can perform the same transition and continue with the updated branches $C_1'$ and $C_2'$, followed by $C$.
Rule \textsc{delay-choice} is similar to \textsc{delay-cond} but for choreographic choices. 

\Cref{chor:semantics:seqcomp-border} shows an important semantics property of choreographies: sequential composition strictly imposes a boundary of the execution of choreographies, which means transitions never cross a sequential composition operator $(\fatsemi)$. This property is essential for distinguishing the \textsc{delay}, \textsc{delay-cond}, and \textsc{delay-choice} cases when reasoning about the behaviour equivalence between out-of-order executions of choreographies and parallel executions of networks in \cref{sec:epp:correctness}.
\begin{theorem}[Execution Does Not Cross Sequential Composition]
\label{chor:semantics:seqcomp-border}
For any choreographic configuration $\langle C_1 \fatsemi C_2, \cstore, \mathscr{C} \rangle$, if it can perform a transition with the transition label $\mu$, where \emph{$\pn(\mu) \subseteq \pn(C_1)$}, then there exists a choreography $C_1'$ such that $\langle C_1, \cstore, \mathscr{C} \rangle \xrightarrow[]{\mu} \langle C_1', \cstore', \mathscr{C} \rangle$ and the resulting choreographic configuration after the $\mu$-transition performed by $\langle C_1 \fatsemi C_2, \cstore, \mathscr{C} \rangle$ is $\langle C_1' \fatsemi C_2, \cstore', \mathscr{C} \rangle$.
\end{theorem}
\begin{proof}
    By structural induction on the derivation of the transition.
\end{proof}

\subsection{Well-Formedness}
\label{chor:wf}
We define \emph{well-formedness} checks for choreographies and configurations to exclude erroneous situations such as a process communicating with itself or incorrect usage of procedure parameters.

We first define the well-formedness judgement for choreographies, written $\langle \mathscr{C},\{\pids{r}\}, C \rangle \checkmark$, which states that a choreography $C$ is well-formed w.r.t. a set of choreographic procedure definitions $\mathscr{C}$ and a set of participating processes $\{\pids{r}\}$.
In particular, each communication or selection must involve two distinct processes from $\{\pids{r}\}$.
For a choreography with procedure calls and/or runtime terms to be well-formed, all processes in the arguments of the procedure call must be distinct and belong to $\{\pids{r}\}$, the invoked procedure is defined, and the number of arguments matches the number of parameters in the invoked procedure.
The detailed definition of $\langle \mathscr{C}, \{\pids{r}\}, C \rangle \checkmark$ is in \cref{app:chor:wf}.
We also give the formal definition of the \emph{well-formedness of choreographic procedure definitions $\mathscr{C}$} as a predicate over $\mathscr{C}$ and the \emph{well-formedness of choreographic configurations} in definitions~\ref{chor:wf:config:def} and \ref{chor:wf:procs:def}.
\begin{definition}[Well-Formedness of Choreographic Procedure Definitions]\label{chor:wf:procs:def}
A set of choreographic procedure definitions $\mathscr{C}$ is well-formed if for each $X(\pids{p}) = C$ in $\mathscr{C}$,
all processes in $\pids{p}$ are distinct, the choreography $C$ does not contain any runtime terms and $\langle \mathscr{C}, \{\pids{p}\}, C \rangle \checkmark$.
\end{definition}

\begin{definition}[Well-Formedness of Choreographic Configurations]\label{chor:wf:config:def}
    A choreographic configuration $\langle C, \cstore, \mathscr{C} \rangle$ is well-formed if both $\mathscr{C}$ and $C$ are well-formed.
\end{definition}

\Cref{chor:wf:preservation:thm} states that well-formedness is preserved by transitions, which is crucial for proving the behavioural equivalence between choreographies and their projections discussed in section~\ref{sec:epp:correctness}.
\begin{theorem}[Preservation of Well-Formedness]\label{chor:wf:preservation:thm}
    Given a well-formed choreographic configuration $\langle C, \cstore, \mathscr{C} \rangle$, if $\langle C, \cstore, \mathscr{C} \rangle \xrightarrow[]{\mu} \langle C', \cstore', \mathscr{C} \rangle$,
    then $\langle C', \cstore', \mathscr{C} \rangle$ is also well-formed.
\end{theorem}
\begin{proof}
    By structural induction on the derivation of the transition.
\end{proof}

\section{Processes and Networks}
\label{sec:proc-net}
Recall the two \emph{networks} we have introduced in examples~\ref{ex:EPPinv} and \ref{ex:EPPinvC}, which are collections of \emph{endpoint programs}.
Endpoint programs are expressed in the \emph{process language}. In this section, we give the formal syntax of the process language and the semantics of networks.

\subsection{Syntax}
\label{sec:proc-net:syntax}
\Cref{proc:syntax:fig} shows the formal syntax of the process language, which shares the same expression $e$ as choreographies. 
Process terms $P$ and instructions $I$ are again mutually defined.
Note that we use $I$ for both choreographic and process instructions. The meaning of $I$ is always clear from the context.

\begin{figure}
\vspace{-1.0em}
\begin{spreadlines}{0pt}
\begin{alignat*}{3}
    \text{Process Procedure Definitions ($\mathfrak{P}$)}\ &&
    \mathscr{P} \metaDeff&\ \{X_i(\pids{p_i}) = P_i\}_{i \in \idx{I}} \\
    \text{Process Term (\setof{Process})}\ &&
    P, Q \metaDeff&\ I \seq P \cmid \precvl p \cmid \pnil\\
    \text{Process Instruction (\setof{PInstr})}\ &&
    I \metaDeff&\ x \metaDef e \cmid \psend{p}{e} \cmid \precv{p}{x} \cmid \psel{p} \cmid \condif e \condthen P\condelse Q \\ && \cmid &\ P + Q \cmid X(\pids{p})\\
    \text{Expression (\setof{Expr})}\ &&
    e \metaDeff&\ v \cmid x \cmid f(\vec{e})
\end{alignat*}
\end{spreadlines}
\caption{Syntax of the Process Language}
\label{proc:syntax:fig}
\vspace{-1.0em}
\end{figure}

There are three process terms: a sequence of process instructions constructed by the overloaded symbol of sequence ($\!\seq\!$), a branching term $\precvl p$, and a terminated process denoted by the overloaded symbol $\pnil$. We again omit trailing occurrences of $\pnil$ for readability.
Amongst these three process terms, the branching term is most interesting one.
In example~\ref{ex:EPPinvC}, the network for $\pid{q}$ projected from the choreography $\cn{inv}$,
$\atomicn{q}{\hlb{\pid{p}\br}\hlb{\{\lbl{pay}: \precv{p}{x} \sep \lbl{reimburse}: \psend{p}{\fname{coords()}} \seq \precv{p}{x} \sep \lbl{skip}: \pnil \}}}$,
contains a highlighted branching term, which specifies three possible selection labels -- $\lbl{pay}$, $\lbl{reimburse}$, and $\lbl{skip}$ -- that $\pid q$ may receive from $\pid p$. Upon receiving a label, $\pid q$ proceeds with the corresponding branch.
Formally, a branching term specifies that a process receives a selection label $\albl_i$ from $\pid{p}$ and continues with the corresponding process term $P_i$. Specifically, $\{\albl_i : P_i\}_{i\in \idx{I}}$ is a \emph{finite map} that associates each distinct label $\albl_i$ to a process term $P_i$, where $\idx{I}$ is a finite index set.
In Lean~4, we represent branching terms as \emph{sorted lists} of label--process term pairs, ordered strictly by their labels.
The list representation allows an efficient merge sorted based implementation of the \emph{merge operation} for projecting branching terms from choreographies with selections (cf. \cref{epp:merge}) and facilitates inductive reasoning for the properties of merge and EPP (cf. \cref{epp:projection}).

There are four basic instructions: \emph{assignment}, the \emph{send}, the \emph{receive}, and the \emph{selection}.
An assignment $x \metaDef e$ stores the result of evaluating an expression $e$ into a local variable $x$.
A send instruction $\psend{p}{e}$ and a receive instruction $\precv{p}{x}$ are dual: the former sends the result of evaluating an expression $e$ to a process $\pid{p}$, and the latter receives a value from a process $\pid{p}$ and stores it into a local variable $x$.
In example~\ref{ex:EPPinv}, the network for $\pid p$ contains $\precv{q}{c}$, which receives the coordinates from $\pid q$ and stores them in a local variable $c$, while the network for $\pid q$ contains a dual instruction $\psend{p}{\fname{coords()}}$, sending the coordinates to $\pid p$.
A selection $\psel{p}$ sends a selection label $\albl$ to a process $\pid{p}$, which is the dual of a branching term. In example~\ref{ex:EPPinvC}, $\pselg{q}{pay}$ in the network for $\pid p$ sends the label $\lbl{pay}$ to $\pid q$.

Similar to choreographies, we also have instructions for making choices between alternative process terms, which are the \emph{conditional} construct (in example~\ref{ex:EPPinvC}) and the \emph{nondeterministic choice}.
A conditional $\condif e \condthen P \condelse Q$ proceeds with $P$ if the expression $e$ evaluates to $\mathit{true}$, and proceeds with $Q$ otherwise, while a nondeterministic choice $P + Q$ can proceed with either $P$ or $Q$ nondeterministically.
The nondeterministic choice is important for modelling the equivalent behaviour of the choreographic choice in networks.
In example~\ref{ex:net-fcfs}, the network projected from the first-come-first-served choreography $\cn{fcfs}$ in example~\ref{ex:chor-fcfs} uses a nondeterministic choice.
\begin{example}[A First-Come-First-Served Network]\label{ex:net-fcfs}
Here $\pid{p_1}$ and $\pid{p_2}$ execute the same program competing for the computational resource of $\pid s$, whereas $\pid s$ knows which one of them wins the competition and works with the winner by using it as the argument in the \emph{process procedure call} $W_2.$
$W_1$ and $W_2$ are the projections of $W$ w.r.t. to its first and second parameters, respectively (cf. \cref{epp:projection}).
\begin{spreadlines}{0pt}
\begin{align*}
\left.
\begin{aligned}[c]
\atomicn{s}{&\pid{p_1}\br\{\lbl{req}: \pselg{p_1}{ok}\seq W_2(\pid{p_1})\seq\\&\pid{p_2}\br\{\lbl{req}: \pselg{p_2}{ko}\}\} \\ \choice \ &\pid{p_2}\br\{\lbl{req}: \pselg{p_2}{ok}\seq W_2(\pid{p_2})\seq\\&\pid{p_1}\br\{\lbl{req}: \pselg{p_1}{ko}\}\}}
\end{aligned}
\;\;\middle|\;\;
\begin{aligned}[c]
\atomicn{p_1}{&\pselg{s}{req} \seq \\&\pid{s}\br\{\lbl{ok} : W_1(s)\sep \lbl{ko} : \pnil \}}
\end{aligned}
\;\;\middle|\;\;
\begin{aligned}[c]
\atomicn{p_2}{&\pselg{s}{req} \seq \\&\pid{s}\br\{\lbl{ok} : W_1(s)\sep \lbl{ko} : \pnil \}}
\end{aligned}
\right.
\end{align*}
\end{spreadlines}
\end{example}

The conditional and nondeterministic choice in the process language enable the triage choreography $\cn{triage}$ in example~\ref{ex:chor-triage} to be projected into the network shown below:
\begin{example}[A Triage Network]\label{ex:net-triage} The network shows a distributed view of each participant in the triage scenario, combining a nondeterministic choice with conditionals, the doctor $\pid d$ prioritise the decision to call for an urgent surgery based on the X-ray result.
\begin{spreadlines}{0pt}
\begin{align*}
\left.
\begin{aligned}[c]
\atomicn{d}{&\precv{rx}{x}\seq\\&(\condif\fname{bad(x)}\condthen \pselg{op}{on} \\ &\condelse \pselg{op}{no}) \seq \precv{bw}{y}\\\choice\ &\precv{bw}{y}\seq\precv{rx}{x}\seq\\ &(\condif\fname{bad(x)} \condthen \pselg{op}{on} \\&\condelse \pselg{op}{no})}
\end{aligned}
\;\;\middle|\;\;
\begin{aligned}[c]
\atomicn{op}{\pid{d}\br\{\lbl{on}: \pnil \sep \lbl{no}: \pnil\}}
\end{aligned}
\;\;\middle|\;\;
\begin{aligned}[c]
\atomicn{rx}{\psend{d}{\fname{test()}}}
\end{aligned}
\;\;\middle|\;\;
\begin{aligned}[c]
\atomicn{bw}{\psend{d}{\fname{test()}}}
\end{aligned}
\right.
\end{align*}
\end{spreadlines}
\end{example}

Lastly, for modelling general recursion, we have the process procedure call $X(\pids{p})$ to denote the invocation of a procedure $X$ with arguments $\pids{p}$, and a set of \emph{process procedure definitions} $\mathscr{P}$.
\begin{example}[Recursive First-Come-First-Served Process Procedure Definitions]\label{ex:net-fcfs-rec}
Recall the recursive first-come-first-served choreographic procedure definition in example~\ref{ex:chor-fcfs-rec}, where the recursive procedure call takes the form of $X(\pid{s},\pid{p_1},\pid{p_2})$. It is projected to three process procedure definitions, each one of which covers the perspective of one of the process parameters.
The process procedure definition w.r.t. $\pid{p_1}$ is
$X_2(\pid{s},\pid{p_2}) = \pselg{s}{req} \seq \pid{s}\br\{\lbl{ok} : W_1(s)\sep \lbl{ko} : \pnil \} \fatsemi (X_2(\pid{s},\pid{p_2}))$,
and the one w.r.t. $\pid{p_2}$ is
$X_3(\pid{s},\pid{p_1}) = \pselg{s}{req} \seq \pid{s}\br\{\lbl{ko} : \pnil \sep \lbl{ok} : W_1(s)\} \fatsemi (X_3(\pid{s},\pid{p_1}))$. The overloaded $(\fatsemi)$ notation is the sequential composition of processes that will be presented in \cref{sec:proc:semantics:seqcomp}.
The process procedure definition w.r.t. $\pid{s}$ is more complex as $\pid s$ is the process where the choreographic choice is resolved:
\begin{spreadlines}{0pt}
\begin{align*}
X_1(\pid{p_1},\pid{p_2}) = \ &\pid{p_1}\br\{\lbl{req}: \pselg{p_1}{ok}\seq W_2(\pid{p_1})\seq\pid{p_2}\br\{\lbl{req}: \pselg{p_2}{ko}\}\} \\ \choice \ &\pid{p_2}\br\{\lbl{req}: \pselg{p_2}{ok}\seq W_2(\pid{p_2})\seq\pid{p_1}\br\{\lbl{req}: \pselg{p_1}{ko}\}\} \seq X_1(\pid{p_1},\pid{p_2}) 
\end{align*}
\end{spreadlines}
\end{example}
We have briefly mentioned that examples \ref{ex:net-fcfs}, \ref{ex:net-triage}, and \ref{ex:net-fcfs-rec} are \emph{generated} by EPP from their corresponding choreographies and choreographic procedure definitions. We will discuss EPP in \cref{sec:epp}.
Having introduced networks informally through examples, we now define them formally.
\begin{definition}[Network] \label{proc:network:def}
A \emph{network} is a \emph{function} mapping process names to process terms, where the zero element of the codomain \setof{Process} is the terminated process $\pnil$.
\[ 
\mathrm{Network\;(\setof{Net})}\quad N : \setof{PName} \to \setof{Process}
\]
\end{definition}
Since a network is a function, we establish the equality of two networks using \emph{extensional equality}.
Formally, two networks $N$ and $M$ are equal, if for all processes $\pid{p}$, we have $N(\pid{p}) = M(\pid{p})$.
Given a network $N$, we say that a process $\pid p$ is \emph{running} if $\pid p \in \support(N)$ (we defined function support in \cref{sec:finfun}).
Furthermore, if two networks $N$ and $M$ have \emph{disjoint supports}, i.e., $\support(N) \disjoint \support(M)$, then we say that the two networks are \emph{disjoint}, written as $N \disjoint M$.
\begin{figure}
\vspace{-1.0em}
\begin{spreadlines}{0pt}
\begin{align*}
\begin{aligned}[c]
    \nnil (\pid{p}) \triangleq \pnil 
\end{aligned}
\quad\quad
\begin{aligned}[c]
    \pid{p}[P] (\pid {q})&\triangleq 
    \begin{cases}
    \ P &\text{if $\pid{p} = \pid{q}$}\\[-0.5ex]
    \ \pnil &\text{otherwise}
    \end{cases}
\end{aligned}
\quad
\begin{aligned}[c]
    (N \parcomp M)(\pid {p})&\triangleq 
    \begin{cases}
    \ N(\pid{p}) &\text{if $\pid{p} \in \support(N)$}\\[-0.5ex]
    \ M(\pid{p}) &\text{otherwise}
    \end{cases}
\end{aligned}
\end{align*}
\end{spreadlines}
\vspace{-0.5em}
\caption{Terminated and Atomic Networks, and Parallel Composition of Networks}
\label{proc:network:nil:atomic:fig}
\vspace{-1.0em}
\end{figure}

In \cref{proc:network:nil:atomic:fig} we formally define three kinds of networks: \emph{terminated networks}, \emph{atomic networks}, and the \emph{parallel composition} of networks.
We again overload the notation $\nnil$ for a terminated network, which maps all process to the terminated process term.
An atomic network $\pid{p}[P]$ maps $\pid p$ to a non-terminated process term $P$ and every other process to $\pnil$.
We write $N \parcomp M$ for the parallel composition of two networks $N$ and $M$, which is left biased and right associative.
The set of networks together with parallel composition form a \emph{monoid} (cf. proposition~\ref{app:proc:parcomp:prop}).
If $N$ and $M$ are disjoint, then $N \parcomp M$ is commutative, i.e., $N \parcomp M \triangleq M \parcomp N$. Consequently, restricting parallel composition to disjoint networks yields a \emph{partial commutative monoid}, as discussed by~\citet{M23}.

In practice, networks obtained by EPP have finite support,
since choreographies specify the behaviour of a finite set of processes and EPP does not introduce new processes.
Thus, we introduce the notation of a network with finite support, which is a \emph{finite network}, written as $\hat{N}$, in definition~\ref{proc:finnet:def}.
\begin{definition}[Finite Network]\label{proc:finnet:def}
A \emph{finite network} $\hat{N}$ is a \emph{finite function} (introduced in \cref{sec:finfun}) from process names to process terms, i.e.,
$\hat{N} : \textbf{PName} \to_f \textbf{Process}$.
\end{definition}

Lastly, we have the \emph{process removal} and the \emph{network restriction} operations (cf. \cref{proc:network:removal-restriction:def}).
The process removal operation $N \setminus \pid{p}$ updates $N$ by setting $N(\pid{p}) = \pnil$, which is commutative and associative. We can thus generalise it to finite sets of processes, written as $N \setminus \{\pids{p}\}$.
The network restriction operation is the opposite, which extracts a subnetwork containing only a finite set of processes $\{\pids{p}\}$ from a given network $N$.
The resulting subnetwork is a finite network written as $N \restrict \{\pids{p}\}$.
\begin{figure}[b]
\vspace{-1.0em}
\begin{spreadlines}{0pt}
\begin{align*}
\begin{aligned}[c]
    (N \setminus \pid{p})(\pid{q}) &\triangleq 
    \begin{cases}
    \ \pnil &\text{if $\pid{p} = \pid{q}$}\\[-0.5ex]
    \ N(\pid{q}) &\text{otherwise}
    \end{cases}
\end{aligned}
\quad\quad
\begin{aligned}[c]
    (N \restrict \{\pids{p}\})(\pid{p}) &\triangleq 
    \begin{cases}
    \ N(\pid{p}) &\text{if $\pid{p} \in \{\pids{p}\}$}\\[-0.5ex]
    \ \pnil &\text{otherwise}
    \end{cases}
\end{aligned}
\end{align*}
\end{spreadlines}
\caption{Process Removal and Network Restriction}
\label{proc:network:removal-restriction:def}
\vspace{-1.0em}
\end{figure}

\subsection{Semantics}
\label{sec:proc-net:semantics}
Similar to choreographies, we define an LTS semantics for \emph{network configurations} $\langle N, \cstore, \mathscr{P} \rangle$,
with the transition labels used in the LTS of choreographies (\cref{chor:semantics:transition-labels}).
This is important for establishing the behavioural equivalence between choreographies and their projected networks in \cref{sec:epp:correctness}.

\subsubsection{Process Name Substitution}
Similar to modelling the execution of procedure calls in choreographies,
we write $P[\pid{r}/\pid{s}]$ and $I[\pid{r}/\pid{s}]$ for the substitution of all occurrences of $\pid{s}$ with $\pid r$ in a process term $P$ and an instruction $I$ respectively.
We write $P[\pids{r}/\pids{s}]$ as a shorthand for substituting many process names in $P$ all at once, where $\pids{r}$ and $\pids{s}$ have the same length.
Details are shown in definition~\ref{proc:semantics:proc-substitution}.

\subsubsection{Sequential Composition}
\label{sec:proc:semantics:seqcomp}
In \cref{proc:semantics:seqcomp} we formally define the sequential composition operation for process terms and networks,
where we overload the notation $(\fatsemi)$.
Similar to sequential composition for choreographies, it is merely an operation for defining the network semantics and EPP rather than a part of the syntax of process and network languages.
We define sequential composition for networks as the point-wise extension of sequential composition for process terms.
Both of the set of process terms with the sequential composition operator and the set of networks with the sequential composition operator form monoids. Details are given in propositions~\ref{app:proc:seqcomp:prop} and \ref{app:net:seqcomp:prop}.
\begin{figure}
\vspace{-1.0em}
\begin{align*}
\begin{aligned}[c]
\pnil \fatsemi P &\triangleq P
\end{aligned}
\quad\quad\quad
\begin{aligned}[c]
(I\seq P') \fatsemi P &\triangleq I \seq (P' \fatsemi P)
\end{aligned}
\quad\quad\quad
\begin{aligned}[c]
\precvl p \fatsemi P &\triangleq \pid{p} \,\&\,\{\albl_i : P_i \fatsemi P\}_{i \in \idx{I}}
\end{aligned}
\end{align*}
\vspace{-1.7em}
\begin{align*}
(N \fatsemi M)\, (\pid{p}) \triangleq (N\, (\pid{p}))\fatsemi (M \, (\pid{p}))
\end{align*}
\caption{Sequential Composition of Process Terms and Networks}
\label{proc:semantics:seqcomp}
\vspace{-1.0em}
\end{figure}

\subsubsection{LTS for Networks}
We give the formal definition of the LTS for networks in \cref{proc:net:semantics:fig}.
Again, we write only the pre-transition label $\rho$ in the transition rules when the choice map $\zeta$ is empty.

\begin{figure}[b]
\vspace{-1.0em}
\begin{mathpar}
    \inferrule*[right={\scriptsize assign}] { \cstore(\pid{p}) \vdash e \downarrow v }
    {\langle \pid{p}[x \metaDef e \seq P], \cstore, \mathscr{P} \rangle \xrightarrow[]{\tlassign{\pid{p}}} \langle \pid{p}[P], \cstore[\pid{p}.x \mapsto v], \mathscr{P} \rangle}

    \inferrule*[right={\scriptsize com}] { \cstore(\pid{p}) \vdash e \downarrow v } 
    {\langle \pid{p}[\psend{q}{e} \seq P] \parcomp \pid{q}[\precv{p}{x} \seq Q], \cstore, \mathscr{P} \rangle \xrightarrow[]{\pid{p}.v \tto \pid{q}} \langle \pid{p}[P] \parcomp \pid{q}[Q], \cstore[\pid{q}.x \mapsto v], \mathscr{P} \rangle}

    \inferrule*[right={\scriptsize sel}] { j \in \mathcal{I} } 
    {\langle \pid{p}[\pseli{q}{j} \seq P] \parcomp \pid{q}[\precvl p ], \cstore, \mathscr{P} \rangle \xrightarrow[]{\genseli{j}} \langle \pid{p}[P] \parcomp \pid{q}[P_j], \cstore, \mathscr{P} \rangle}

    \inferrule*[right={\scriptsize cond-then}] { \cstore(\pid{p}) \vdash e \downarrow \mathit{true}} 
    {\langle \pid{p}[\condif e \condthen P_1 \condelse P_2 \seq Q], \cstore, \mathscr{P} \rangle \xrightarrow[]{\tlthen{p}} \langle \pid{p} [P_1 \fatsemi Q], \cstore, \mathscr{P} \rangle }

    \inferrule*[right={\scriptsize cond-else}] { \cstore(\pid{p}) \vdash e \downarrow v \\ v \neq \mathit{true}} 
    {\langle \pid{p}[\condif e \condthen P_1 \condelse P_2 \seq Q], \cstore, \mathscr{P} \rangle \xrightarrow[]{\tlelse{p}} \langle \pid{p} [P_2 \fatsemi Q], \cstore, \mathscr{P} \rangle }
    
    \inferrule*[right={\scriptsize choice-l}]{\langle \pid{p}[P] \parcomp N, \cstore, \mathscr{P}\rangle \xrightarrow[]{\mu} \langle \pid{p}[P'] \parcomp N', \cstore', \mathscr{P}\rangle \\ \pid{p} \in \pn(\mu) \\ \pid{p} \notin \support(N) \\ \pid{p} \notin \support(N') \\ \support(N) \cup \{\pid{p}\} = \pn(\mu)}{\langle \pid{p}[P \choice Q \seq R] \parcomp N, \cstore, \mathscr{P}\rangle \xrightarrow[]{\mu.\mathit{addChoice}\; \pid{p}\; \mathit{left}} \langle\pid{p}[P' \fatsemi R] \parcomp N', \cstore', \mathscr{P} \rangle}
    
    \inferrule*[right={\scriptsize choice-r}]{\langle \pid{p}[Q] \parcomp N, \cstore, \mathscr{P}\rangle \xrightarrow[]{\mu} \langle \pid{p}[Q'] \parcomp N', \cstore', \mathscr{P}\rangle \\ \pid{p} \in \pn(\mu) \\ \pid{p} \notin \support(N) \\ \pid{p} \notin \support(N') \\ \support(N) \cup \{\pid{p}\} = \pn(\mu)} {\langle \pid{p}[P \choice Q \seq R] \parcomp N, \cstore, \mathscr{P}\rangle \xrightarrow[]{\mu.\mathit{addChoice}\; \pid{p}\; \mathit{right}} \langle\pid{p}[Q' \fatsemi R] \parcomp N', \cstore', \mathscr{P} \rangle}

    \inferrule*[right={\scriptsize choice-skip}] { } 
    {\langle \pid{p}[\cnil + \cnil \seq P], \cstore, \mathscr{P} \rangle \xrightarrow[]{\tlskip{p}} \langle \pid{p} [P], \cstore, \mathscr{P} \rangle }

    \inferrule*[right={\scriptsize call}] {X(\pids{r}) = Q \in \mathscr{P}} 
    {\langle \pid{p}[X(\pids{q}) \seq P], \cstore, \mathscr{P} \rangle \xrightarrow[]{\tlcall{p}} \langle \pid{p} [Q[\pids{q}/\pids{r}] \fatsemi P], \cstore, \mathscr{P} \rangle }
    \quad\;
    \inferrule*[right={\scriptsize par}] {\langle N, \cstore, \mathscr{P} \rangle \xrightarrow[]{\mu} \langle N', \cstore', \mathscr{P} \rangle \\ N \disjoint M} 
    {\langle N \parcomp M, \cstore, \mathscr{P} \rangle \xrightarrow[]{\mu} \langle N' \parcomp M, \cstore', \mathscr{P} \rangle }
\end{mathpar}
\caption{Semantics of Networks}
\label{proc:net:semantics:fig}
\vspace{-1.0em}
\end{figure}

As distributed implementations of choreographies, we have three basic rules for local assignments, communications, and selections in a network.
Rule \textsc{assign} models the execution of an assignment on the process $\pid{p}$, where the result of evaluating an expression $e$ in the local store of $\pid{p}$, namely $v$, is stored into $\pid p$'s local variable $x$.
Rule \textsc{com} models the communication between two processes $\pid{p}$ and $\pid{q}$ in a network, where $\pid{p}$ evaluates an expression $e$ in its local store and sends the result to $\pid{q}$. Upon receipt of the value (from $\pid p$), $\pid{q}$ stores it in its local variable $x$.
Rule \textsc{sel} is similar to the rule \textsc{com} but for communicating selection labels.
The network contains a selection instruction at $\pid{p}$ and a compatible branching term at $\pid{q}$ that performs a transition in which $\pid{p}$ sends a selection label $\albl_j$ to $\pid{q}$, and upon receipt of $\albl_j$, $\pid q$ continues with the process term $P_j$ associated with $\albl_j$.

The execution of a conditional in networks is similar to choreographies but from the perspective of a guard process $\pid p$, modelled by the rules \textsc{cond-then} and \textsc{cond-else}.
Rule \textsc{cond-then} states that, if a guard expression $e$ is evaluated to $\mathit{true}$ in the local store of $\pid{p}$, then the network performs a transition with label $\tlthen{p}$ and continues with the then-branch $P_1$, followed by the continuation $Q$ which is appended to $P_1$ by the sequential composition operator $(\fatsemi)$.
Rule \textsc{cond-else} instead continues with the else-branch $P_2$ if the expression $e$ is not evaluated to $\mathit{true}$.
Since the resolution of a nondeterministic choice is decided by the system rather than a single process, processes other than the one at which the choice is resolved may also be involved.
Example~\ref{ex:net-nested-choice} shows the transition of the network resulting from the EPP of the nested choreographic choices $\cn{pq}$ and $\cn{qp}$ in example~\ref{ex:chor-nested-choice}:
\begin{example}[Nondeterministic Choice]\label{ex:net-nested-choice} In the network $\atomicn{p}{\psend{q}{1} \choice \psend{q}{2}} \parcomp \atomicn{q}{\precv{p}{x} + \precv{p}{y}}$, resolving the choice at $\pid p$ requires interacting with $\pid q$ and vice versa. 
If the choice is resolved to the left branch at $\pid{p}$ and the right branch at $\pid{q}$, then the network can perform a transition by applying rule \textsc{choice-l} to $\pid{p}$ and rule \textsc{choice-r} to $\pid{q}$ in any order such that $\pid p$ sends the value $1$ to $\pid q$, which stores it at $y$:
    {
    \[ \atomicn{p}{\psend{q}{1} \choice \psend{q}{2}} \parcomp \atomicn{q}{\precv{p}{x} + \precv{p}{y}} \xrightarrow[]{\addChoiceg{(\addChoiceg{\tlcom{p}{1}{q}}{p}{left})}{q}{right}} \atomicn{p}{\pnil} \parcomp \atomicn{q}{\pnil} \]
    }
\end{example}
Formally, rules \textsc{choice-l} and \textsc{choice-r} modelling the execution of a nondeterministic choice include not only the process $\pid{p}$, but also a network $N$ involved in resolving the choice.
The requirements $\pid{p} \in \pn(\mu)$, $\pid{p} \notin \support(N)$, and $\support(N) \cup \{\pid{p}\} = \pn(\mu)$ ensure that the network $N$ is the minimal network that is involved in transition for resolving the choice.
When a choice is resolved without involving any process other than $\pid{p}$, for instance, $\pid{p}[ x \metaDef 1 \choice x \metaDef 2]$, the network $N$ in the rules can simply be $\nnil$.
Rule \textsc{choice-l} states that if the process term $P$ can perform a $\mu$-transition to $P'$ and the rest of the network $N$ can perform a compatible $\mu$-transition to $N'$, then the choice at $\pid{p}$ is resolved to the left branch by the transition, and $\pid p$ continues with the process term $P' \fatsemi R$.
Rule \textsc{choice-r} is symmetric to \textsc{choice-l} where the choice at $\pid{p}$ is resolved to the right branch $Q$.
We also have a garbage collection rule \textsc{choice-skip} corresponds to the rules \textsc{skip-l} and \textsc{skip-r} in the semantics of choreographies:
if the process $\pid{p}$ has a choice between two $\pnil$s, then $\pid p$ skips $\pnil \choice \pnil$ and proceeds with the continuation $P$. Example~\ref{ex:net-gc} illustrates the need of the garbage collection rule.
\begin{example}[A Nondeterministic Choice Needs the Garbage Collection Rule]\label{ex:net-gc}
The network resulting from the EPP of $\cn{gc}$ in example~\ref{ex:chor-gc} is $\atomicn{q}{\pnil \choice \pnil \seq \precv{p}{y}} \parcomp \atomicn{p}{x \metaDef 42 \seq \psend{q}{42}}$. It requires rule \textsc{choice-skip} together with rule \textsc{par} introduced later to get rid of the unmeaningful choice in $\pid q$:
{
\[ \atomicn{q}{\pnil \choice \pnil \seq \precv{p}{y}} \parcomp \atomicn{p}{x \metaDef 42 \seq \psend{q}{42}} \xrightarrow[]{\tlskip{q}} \atomicn{q}{\precv{p}{y}} \parcomp \atomicn{p}{x \metaDef 42 \seq \psend{q}{42}}\]
}
\end{example}

Rule \textsc{call} models recursions.
It states that if there is a procedure definition $X(\pids{r}) = Q$ in the set of procedure definitions $\mathscr{P}$,
then $\pid{p}$ can execute the procedure call $X(\pids{q})$ by substituting the $\pids{q}$ for the formal parameters $\pids{r}$ in the body of the procedure $Q$, followed by the continuation $P$.
Rule \textsc{par}
models the parallel execution of independent transitions: for two disjoint networks $N$ and $M$,
if $N$ can perform a $\mu$-transition, then $N \pp M$ can perform the same transition leaving $M$ unaffected.

\section{Endpoint Projection (EPP)}
\label{sec:epp}
As we have demonstrated in previous sections, choreographies are compiled into networks by EPP.
A key ingredient in the design of EPP is handling \emph{knowledge of choice (KoC)} correctly, as we have briefly discussed in examples~\ref{ex:CinvC} and \ref{ex:EPPinvC}.
Thus, before formally defining EPP, we introduce the \emph{merge operation} ($\sqcup$) for process terms which takes care of KoC.

\subsection{Merge and the Branching Order}
\label{epp:merge}
In the literature~\cite{M23}, the merge operation $(\sqcup)$ is defined as a \emph{partial function} of the type signature: $\textbf{Process} \rightarrow \textbf{Process} \rightarrow \opt{Process}$,
where $\opt{Process}$ is an option type. It takes two process terms and returns a ``some” term $\some{P}$ if they are indeed mergeable, where $P$ is the resulting process term; otherwise, it is undefined and returns a ``none'' term represented as $\none$.
Moreover, $(\textbf{Process}, \sqcup)$ has been claimed~\cite{M23} to form a \emph{partial semilattice}, meaning that properties of semilattices is satisfied when the merge operation is defined.

In our work, we give a more precise definition to the merge operation, where the merge operation is mutually defined for process terms and instructions with input process terms and instructions being option types as well. Thus, we can define the merge operation as a total function and obtain \emph{join-semilattice} structures $(\opt{Process}, \sqcup)$ and $(\opt{PInstr}, \sqcup)$. The type signatures are given below:
\begin{spreadlines}{0pt}
\begin{align*}
\begin{aligned}[c]
    \sqcup : \opt{Process}\rightarrow \opt{Process} \rightarrow \opt{Process}
\end{aligned}
\quad\quad\quad
\begin{aligned}[c]
    \sqcup : \opt{PInstr}\rightarrow \opt{PInstr} \rightarrow \opt{PInstr}
\end{aligned}
\end{align*}
\end{spreadlines}
\\[-4.0ex]
We formally define $(\sqcup)$ for process terms and instructions in \cref{merge:proc:fig} and \cref{merge:instr:fig}. Note that $\opterm{P}$ is a term of the option type $\opt{Process}$, i.e., either $\some{P}$ or $\none$, and $\opterm{I}$ is a term of $\opt{PInstr}$.
\begin{figure}
\vspace{-1.0em}
\begin{subfigure}{1.0\textwidth}
\begin{spreadlines}{0pt}
\begin{align*}
\begin{aligned}[c]
    \some{\pnil} \sqcup \some{\pnil}\ \triangleq\ \some{\pnil}
\end{aligned}
\quad
\begin{aligned}[c]
    \some{I_1 \seq P_1} \sqcup \some{I_2 \seq P_2} \ \triangleq
    \begin{cases}
    \;\some{I \seq P} &\text{if $\some{I_1} \sqcup \some{I_2} \eqw \some{I}$ and $\some{P_1} \sqcup \some{P_2} \eqw\some{P}$} \\[-0.5ex]
    \;\none &\text{otherwise}
    \end{cases}
\end{aligned}
\end{align*}
\vspace{-1.0em}
\end{spreadlines}
\begin{spreadlines}{0pt}
\begin{align*}
    &\some{\precvl {p_1}} \sqcup \some{\precvlj {p_2}} \ \triangleq\\
    &\begin{cases}
    \some{\{\albl_k \!:\! R_k\}_{k\in \idx{I} \cap \idx{J}} \!\cup\! \{\albl_i \!:\! P_i\}_{i\in \idx{I} \setminus \idx{J}} \!\cup\! \{\albl_j \!:\! Q_j\}_{j\in \idx{J} \setminus \idx{I}}}  &\!\!\!\!\text{if $\pid{p_1}\!=\!\pid{p_2}$, for each $k\!\in\! \idx{I} \!\cap\! \idx{J},\some{P_k} \!\sqcup\! \some{Q_k} \eqw \some{R_k}$} \\[-0.5ex]
    \none &\!\!\!\!\text{otherwise}
    \end{cases}
\end{align*}
\end{spreadlines}
\vspace{-1.0em}
\begin{align*}
    \opterm{P_1} \sqcup \opterm{P_2} \triangleq\ \none \quad\text{if $\opterm{P_1}$ and $\opterm{P_2}$ do not match the above forms}
\end{align*}
\vspace{-1.5em}
\caption{The Merge Operation for Process Terms}
\label{merge:proc:fig}
\end{subfigure}
\begin{subfigure}{1.0\textwidth}
\begin{spreadlines}{0pt}
\begin{align*}
    &\some{\condif e_1 \condthen P_1 \condelse P_2} \sqcup \some{\condif e_2 \condthen Q_1 \condelse Q_2} \;\triangleq\\
    &\begin{cases}
    \;\some{\condif e_1 \condthen R_1 \condelse R_2} \quad &\text{if $e_1 = e_2$ and $\some{P_1} \sqcup \some{Q_1} \eqw \some{R_1}$ and $\some{P_2} \sqcup \some{Q_2} \eqw \some{R_2}$} \\[-0.5ex]
    \;\none &\text{otherwise}
    \end{cases}
\end{align*}
\end{spreadlines}
\vspace{-1.0em}
\begin{spreadlines}{0pt}
\begin{align*}
    \some{P_1 \choice P_2} \sqcup \some{Q_1 \choice Q_2} &\triangleq
    \begin{cases}
    \;\some{R_1 \choice R_2} \quad &\text{if $\some{P_1} \sqcup \some{Q_1} \eqw \some{R_1}$ and $\some{P_2} \sqcup \some{Q_2} \eqw \some{R_2}$} \\[-0.5ex]
    \;\none &\text{otherwise}
    \end{cases}\\[0.2em]
    \some{I} \sqcup \some{I} &\triangleq\ \some{I} \quad\text{if $I$ is not a conditional or a choice}\\[-0.5em]
    \opterm{I_1} \sqcup \opterm{I_2} &\triangleq\ \none \quad\;\text{if $\opterm{I_1}$ and $\opterm{I_2}$ do not match the above forms}
\end{align*}
\end{spreadlines}
\vspace{-1.5em}
\caption{The Merge Operation for Process Instructions}
\label{merge:instr:fig}
\vspace{0.5em}
\end{subfigure}
\caption{The Merge Operation for Process Terms and Process Instructions}
\label{merge:fig}
\vspace{-1.0em}
\end{figure}
The merge of two terminated processes is a terminated process.
The merge of two process terms with the form of $I \seq P$ returns the merge of the instructions and the merge of the continuations chained up by $(\!\seq\!)$, if both are mergeable, and returns a none term $\none$ otherwise.
The merge of two branching terms is the most interesting case of the merge operation.
If the two branching terms have the same process name, their merge contains the union of their selection labels.
For each common label $\albl_k$ with $k \in \idx{I} \cap \idx{J}$, the associated continuation is $P_k \sqcup Q_k$, provided they are mergeable.
Labels appearing in only one branching term
remain associated with their corresponding process terms.
If the process names differ ($\pid{p_1} \neq \pid{p_2}$), or any $P_k \sqcup Q_k$ results in $\none$, then the result is $\none$.
As discussed in \cref{sec:proc-net:syntax}, we represent the branching terms as sorted lists in Lean 4, enabling the merge operation to be implemented with a merge sort algorithm efficiently.
We are thus able to formulate and reason about properties of merge with straightforward inductions.

The merge of two instructions is mostly straightforward.
For instructions other than conditionals and nondeterministic choices, the merge results in the instruction itself if the two instructions are identical, and $\none$ otherwise.
For conditionals, if the two conditionals have the same guard expression, i.e., $e_1 = e_2$, and their corresponding branches are mergeable, i.e., $\some{P_1} \sqcup \some{Q_1} \eqw \some{R_1}$ and $\some{P_2} \sqcup \some{Q_2} \eqw \some{R_2}$, then their merge is $\condif e_1 \condthen R_1 \condelse R_2$. Otherwise, the merge is $\none$. The merge operation for nondeterministic choice is defined similarly: the corresponding branches are merged whenever they are mergeable, and the result is $\none$ otherwise.

Since the merge operation is defined mutually over process terms and instructions, all associated definitions and properties apply to both. Their proofs likewise proceed by mutual induction on the structure of process terms and instructions. For brevity, we present only the results for process terms, as those for instructions follow in the same way.

The set $\opt{Process}$ together with the merge operation, $(\opt{Process}, \sqcup)$, forms a \emph{join-semilattice}, where $(\sqcup)$ is the \emph{join operator}. Thus, the merge operation is idempotent, commutative, and associative (cf. proposition~\ref{app:merge:semilattice:prop}).
In Lean 4, 
we instantiate $(\opt{Process}, \sqcup)$ to \code{SemilatticeSup}, which is the the join-semilattice implementation in Mathlib~\cite{mathlib2020}. This allows us to reuse Mathlib's existing theories of join-semilattices directly in our mechanisation.

To reason about the correctness of EPP, we define a merge operation for finite networks -- here finiteness is particularly important to guarantee that this operation is decidable.
Formally, we define the merge operation for two finite networks $\hat{N}$ and $\hat{M}$ by lifting the merge operation defined for $\opt{Process}$, overloading the notation merge operator $(\sqcup)$, with the type signature:
\[ \sqcup : \setof{FinNet} \rightarrow \setof{FinNet} \rightarrow \opt{FinNet} \]
Such a lifting is defined as a general function $\fname{mapBin}$ for finite functions. Given a binary operation $\op$ defined on the set $\optg{B}$,
we can lift $\op$ to finite functions $\hat{f}$ and $\hat{g}$ of the type $ A \finfun B$:
\begin{align*}
    \fname{mapBin} \;\hat{f}\; \hat{g}\; \op \triangleq
    \begin{cases}
    \some{a \mapsto b_a} &\!\!\!\!\text{if $\support(\hat{f}) = \support(\hat{g})$ and for each $a \in \support(\hat{f}), \,\hat{f}(a) \op \hat{g}(a) =\,\some{b_a} $} \\[-0.5ex]
    \none &\!\!\!\!\text{otherwise}
    \end{cases}
\end{align*}
We then define the merge operation for finite networks as $\hat{N} \sqcup \hat{M} \triangleq \fname{mapBin}\; \hat{N}\; \hat{M}\; \sqcup$, where the $(\sqcup)$ in $\fname{mapBin}\; \hat{N}\; \hat{M}\; \sqcup$ is the merge for process terms, served as $\op$ in the definition of $\fname{mapBin}$.

Utilising the merge operation, definition~\ref{merge:mbrs:def} gives an ordering relation on $\opt{Process}$, namely a \emph{more branches relation} ($\mbrs$).
Intuitively, $P \mbrs Q$ means that if neither $P$ nor $Q$ contains any branching terms, then $P$ and $Q$ are identical. Otherwise, each branching term in $P$ must contain at least as many labels as its corresponding branching term in $Q$.
\begin{definition}[$\,\mbrs$ on \emph{$\opt{Process}$}] \label{merge:mbrs:def}
    For any $\opterm{P}$ and $\opterm{Q}$ of the type $\opt{Process}$, $\opterm{P} \mbrs \opterm{Q}$ iff $\opterm{P} \sqcup \opterm{Q} = \opterm{P}$.
\end{definition}

With the join-semilattice construction, we can derive that the relation $(\mbrs)$ is a \emph{partial order} on $\opt{Process}$, which means that it satisfies the properties of reflexivity, antisymmetry, and transitivity 
(cf. proposition~\ref{app:merge:mbrs:prop}). In Lean 4, we can again obtain these properties from Mathlib.
We also discover an additional algebraic property that connects merge $(\sqcup)$ and sequential composition $(\fatsemi)$: under $\mbrs$, $\sqcup$ and $\fatsemi$ exchange as formally stated in proposition~\ref{merge:seqcomp:mbrs}. we lift $(\fatsemi)$ to $\opt{Process}$, if either $\opterm{P} $ or $\opterm{Q}$ is $\none$, then $\opterm{P} \fatsemi \opterm{Q}$ results in $\none$, otherwise, if both $\opterm{P}$ and $\opterm{Q}$ are $\some{P}$ and $\some{Q}$, then $\opterm{P} \fatsemi \opterm{Q}$ is $\some{P \fatsemi Q}$.
\begin{proposition} \label{merge:seqcomp:mbrs}
For any terms $\opterm{P_1}$, $\opterm{P_2}$, $\opterm{Q_1}$, and $\opterm{Q_2}$ in the set \emph{$\opt{Process}$}:
\[
(\opterm{P_1} \sqcup \opterm{P_2}) \fatsemi (\opterm{Q_1} \sqcup \opterm{Q_2}) \mbrs (\opterm{P_1} \fatsemi \opterm{Q_1}) \sqcup (\opterm{P_2} \fatsemi \opterm{Q_2})
\]
\end{proposition} 
We also define the more branches relation for networks by pointwise lifting the more branches relation on $\opt{Process}$ to the set of networks \setof{Net}, shown in definition~\ref{merge:mbrs:net:def}.
\begin{definition}[$\,\mbrs$ on \emph{\setof{Net}}]\label{merge:mbrs:net:def}
    For any $N$ and $M$ in $\setof{Net}$, $N \mbrs M$ iff for any process $\pid{p}$, $\some{N (\pid{p})} \mbrsw \some{M (\pid{p})}$.
\end{definition}
Intuitively, $N \mbrs M$ means either $N$ and $M$ are identical, or for each process $\pid{p}$, the process term of $\pid{p}$ in $N$ has at least the same number of branches as the corresponding process term of $\pid{p}$ in $M$, provided that both process terms contain branching terms.
The relation $(\mbrs)$ is preserved by network transitions, as formalised in \cref{merge:mbrs:net:transition}. Specifically, if $N \mbrs M$ and the smaller network $M$ performs a transition to $M'$, then $N$ can perform the same transition to some $N'$ such that $N' \mbrs M'$.
\begin{theorem}\label{merge:mbrs:net:transition}
    For any networks $N$ and $M$, if $N \mbrs M$ and $\langle M, \cstore, \mathscr{P} \rangle \xrightarrow[]{\mu} \langle M', \cstore', \mathscr{P} \rangle$, then there exists an $N'$ such that $\langle N, \cstore, \mathscr{P} \rangle \xrightarrow[]{\mu} \langle N', \cstore', \mathscr{P} \rangle$ and $N' \mbrs M'$.
\end{theorem}
\begin{proof}
   By structural induction on the derivation of the transition $\langle M, \cstore, \mathscr{P} \rangle \xrightarrow[]{\mu} \langle M', \cstore', \mathscr{P} \rangle$
\end{proof}

We establish a key property relating network transitions and the merge operation, formalised in \cref{merge:net:transition}. If $N$ and $M$ can both perform a transition with label $\mu$ to $N'$ and $M'$, respectively, and both $N$ and $M$ as well as $N'$ and $M'$ are mergeable on all processes in $\mu$, then the merge of $N$ and $M$ can perform a transition with label $\mu$ to a network related by $(\mbrs)$ to the merge of $N'$ and $M'$.
\begin{theorem}\label{merge:net:transition}
    For any networks $N$, $M$, $N'$, and $M'$ such that $N \restrict {\emph{\pn}(\mu)} \sqcup M \restrict {\emph{\pn}(\mu)} \eqr \some{\hat{W}}$ and $N'\restrict {\emph{\pn}(\mu)} \sqcup M'\restrict {\emph{\pn}(\mu)} \eqr \some{\hat{W'}}$,
    if $\langle N, \cstore, \mathscr{P} \rangle \xrightarrow[]{\mu} \langle N', \cstore', \mathscr{P} \rangle$ and $\langle M, \cstore, \mathscr{P} \rangle \xrightarrow[]{\mu} \langle M', \cstore', \mathscr{P} \rangle$, then there exists a network $W''$ such that $\langle W, \cstore, \mathscr{P} \rangle \xrightarrow[]{\mu} \langle W'', \cstore', \mathscr{P} \rangle$ and $W'' \mbrs W'$. Note that $W$ and $W'$ are the lifted networks from finite networks $\hat{W}$ and $\hat{W'}$, respectively.
\end{theorem}
\begin{proof}
    By induction on the transition label $\mu$ using proposition~\ref{merge:seqcomp:mbrs}.
\end{proof}

\subsection{Process Projection and EPP}
\label{epp:projection}
We define EPP using the \emph{process projection}, which is a function of the type:
\[
\llbracket \cdot \rrbracket\_ : \textbf{Chor} \rightarrow \textbf{PName} \rightarrow \opt{\textbf{Process}}
\]
It is formally defined in \cref{epp:proc-projection:fig}. We say a choreography $C$ is \emph{projectable} onto a process $\pid{r}$ if $\proj{C}{r} \eqr\some{P}$ for some process term $P$, and we say a choreography $C$ is \emph{not projectable} onto a process $\pid{r}$ if $\proj{C}{r} \eqr\none$.

\begin{figure}
\vspace{-1.0em}
\begin{spreadlines}{1pt}
\begin{align*}
\begin{aligned}[c]
    \proj{\cnil}{r} &\triangleq \some{\pnil}
\end{aligned}
\quad\quad
\begin{aligned}[c]
    \proj{\pid{p}.x \metaDef e \seq C}{r} &\triangleq
    \begin{cases}
    \;\some{x \metaDef e \seq P} &\text{if $\pid{r} = \pid{p}$ and $\proj{C}{r} \eqr\some{P}$}\\[-0.5ex]
    \;\proj{C}{r} &\text{otherwise}
    \end{cases}
\end{aligned}
\end{align*}
\end{spreadlines}
\vspace{-1.0em}
\begin{spreadlines}{1pt}
\begin{align*}
    \proj{\gencom \seq C}{r} &\triangleq
    \begin{cases}
    \;\some{\psend{q}{e} \seq P} &\text{if $\pid{r} = \pid{p}$ and $\proj{C}{r} \eqr\some{P}$}\\[-0.5ex]
    \;\some{\precv{p}{x} \seq P} &\text{if $\pid{r} = \pid{q}$ and $\proj{C}{r} \eqr\some{P}$}\\[-0.5ex]
    \;\proj{C}{r} &\text{otherwise}
    \end{cases}\\
    \proj{\gensel \seq C}{r} &\triangleq
    \begin{cases}
    \;\some{\psel{q} \seq P} &\text{if $\pid{r} = \pid{p}$ and $\proj{C}{r} \eqr\some{P}$}\\[-0.5ex]
    \;\some{\precvlpr{p}{P}} &\text{if $\pid{r} = \pid{q}$ and $\proj{C}{r} \eqr\some{P}$}\\[-0.5ex]
    \;\proj{C}{r} &\text{otherwise}
    \end{cases}
\end{align*}
\end{spreadlines}
\vspace{-1.0em}
\begin{spreadlines}{1pt}
\begin{align*}
    &\proj{\condif\pid{p}.e\condthen C_1\condelse C_2 \seq C}{r} \,\triangleq\\
    &\begin{cases}
    \;\some{\condif e\condthen P_1 \condelse P_2 \seq P} &\text{if $\pid{r} = \pid{p}$ and $\proj{C_1}{r} \eqr\some{P_1}$ and $\proj{C_2}{r} \eqr\some{P_2}$ and $\proj{C}{r} \eqr\some{P}$}\\[-0.5ex]
    \;\some{Q \fatsemi P} &\text{if $\pid{r} \neq \pid{p}$ and $\proj{C_1}{r} \sqcup \proj{C_2}{r} \eqr\some{Q}$ and $\proj{C}{r} \eqr\some{P}$}\\[-0.5ex]
    \;\none &\text{otherwise}
    \end{cases}\\
    &\proj{C_1 \choice_\pid{p} C_2 \seq C}{r} \,\triangleq
    \begin{cases}
    \;\some{P_1 + P_2 \seq P} &\text{if $\pid{r} = \pid{p}$ and $\proj{C_1}{r} \eqr\some{P_1}$ and $\proj{C_2}{r} \eqr\some{P_2}$ and $\proj{C}{r} \eqr\some{P}$}\\[-0.5ex]
    \;\some{Q \fatsemi P} &\text{if $\pid{r} \neq \pid{p}$ and $\proj{C_1}{r} \sqcup \proj{C_2}{r} \eqr\some{Q}$ and $\proj{C}{r} \eqr\some{P}$}\\[-0.5ex]
    \;\none &\text{otherwise}
    \end{cases}\\
    &\proj{X(\pids{p}) \seq C}{r} \,\triangleq
    \begin{cases}
    \;\some{X_i(\pids{p}\backslash \pid{r}) \seq P} &\text{if $\pid{r}=\pid{p}_i$ for $\pids{p} = \pid{p}_1,...,\pid{p}_n$ and $1\leq i\leq n$ and $\proj{C}{r} \eqr\some{P}$}\\[-0.5ex]
    \;\proj{C}{r} &\text{otherwise}
    \end{cases}\\
    &\proj{\pid{q} : X(\pids{p}). C' \seq C}{r} \,\triangleq\\
    &\begin{cases}
    \;\some{X_i(\pids{p}\backslash \pid{r}) \seq P'} &\text{if $\pid{r} = \pid{q}$ and $\pid{r}=\pid{p}_i$ for $\pids{p} = \pid{p}_1,...,\pid{p}_n$ and $1\leq i\leq n$ and $\proj{C'}{r} \eqr\some{P'}$}\\[-0.5ex]
    \;\llbracket C \rrbracket_\pid{r} &\text{otherwise}
    \end{cases}
\end{align*}
\end{spreadlines}
\vspace{-0.5em}
\caption{Process Projection of Choreographies}
\label{epp:proc-projection:fig}
\vspace{-1.0em}
\end{figure}
A terminated choreography is projected to a terminated process.
The projection of a choreography of the form $I \seq C$ proceeds by projecting the choreographic instruction $I$.
When $I$ is a local assignment $\pid{p}.x \metaDef e$,
if $\pid{r} = \pid{p}$ and there exists a process term $P$ such that the continuation $C$ projects onto $\pid{r}$ as $\some{P}$, then the projection results in the assignment $x \metaDef e$ followed by $P$. Otherwise, it results in the projection of $C$ onto $\pid r$.
When $I$ is a communication, there are three cases for the projection of $\gencom \seq C$ onto $\pid{r}$. If $\pid{r}$ is the sender $\pid{p}$ and the continuation $C$ is projectable onto $\pid{r}$ resulting in a process term $P$, then the resulting process term is a send instruction $\psend{q}{e}$ followed by $P$. If $\pid{r}$ is the receiver $\pid{q}$ and $C$ is projectable onto $\pid{r}$ resulting in $P$, then the resulting process term is a receive instruction $\precv{p}{x}$ followed by $P$. Otherwise, the result is the projection of $C$ onto $\pid r$.
When $I$ is a selection $\gensel$, the projection of the choreography onto $\pid r$ is similar to the communication case, but results in a selection instruction $\psel{q}$ followed by continuation $P$ and a branching term $\precvlpr{p}{P}$ for the first two cases, respectively.

The merge operation is the key of handling KoC correctly in the definition of the projection onto a process $\pid r$ for conditionals and choreographic choices.
When $I$ is a conditional $\condif\pid{p}.e\condthen C_1\condelse C_2$, we define the projection by analysing if $\pid{r}$ is the guard process $\pid{p}$ or not.
If $\pid r = \pid p$, and both branches $C_1$ and $C_2$ as well as the continuation $C$ are projectable onto $\pid{r}$ resulting in process terms $P_1$, $P_2$, and $P$ respectively, then the result is a conditional of which the guard expression is $e$ and two branches are $P_1$ and $P_2$, followed by $P$.
If $\pid r \neq \pid p$, while the projections of two branches $C_1$ and $C_2$ onto $\pid{r}$ are mergeable, resulting in $Q$, and $C$ is projectable onto $\pid{r}$, resulting in $P$, then the result is a sequential composition of $Q$ and $P$.
Otherwise, the conditional is not projectable on $\pid{r}$, yielding $\none$.
The process projection for choice $C_1 \choice_\pid{p} C_2 \seq C$ is defined similarly to the projection of a conditional by analysing whether or not $\pid{r}$ is the process at which the choreographic choice is resolved.

Lastly, both of choreographic procedure calls and runtime terms are projected to process procedure calls.
For the projection of $X(\pids{p}) \seq C$ onto $\pid{r}$, if $\pid{r}$ is the $i$-th argument in $\pids{p}$ and the continuation $C$ is projectable on $\pid{r}$ resulting in $P$, then the result is the process procedure call $X_i(\pids{p} \setminus {\pid{r}})$ followed by $P$, where $X_i$ represents the procedure from the perspective of the $i$-th argument. Otherwise, the projection is that of $C$.
For the projection of $\pid{q} : X(\pids{p}). C' \seq C$ onto $\pid{r}$, if $\pid{r}$ is the process yet to enter the procedure call ($\pid{r} = \pid{q}$), $\pid{r}$ is $i$-th element of the argument list $\pids{p}$, and the body $C'$ of is projectable on $\pid{r}$, yielding $P'$, then the result is the process procedure call $X_i(\pids{p}\backslash\pid{r})$ followed by $P'$. Here, the continuation is the projection of $C'$ rather than $C$, since $\pid{r}$ has not yet entered the procedure call and therefore should not project beyond it. Otherwise, the result is the projection of $C$ onto $\pid r$.

Proposition~\ref{epp:proc-projection:seqcomp} shows an important algebraic property of process projection: the distributivity of process projection over sequential composition. This exemplifies yet another algebraic property formulated in \mech to improve modularity, which has been left informal in the literature.
\begin{proposition}\label{epp:proc-projection:seqcomp}
    For any choreography $C_1$ and $C_2$, and any process $\pid{p}$:
    $\proj{C_1 \fatsemi C_2}{\pid{p}} = \proj{C_1}{\pid{p}} \fatsemi \proj{C_2}{\pid{p}}$.
\end{proposition}

We define EPP of choreographies as a function of the type: $\epp{\cdot} : \setof{Chor} \to \opt{Net}$, using the process projection.
The EPP of a choreography $C$ can either result in a network mapping each process in $\pn(C)$ to its projected process term and all processes not in $\pn(C)$ to $\pnil$, if $C$ is projectable on all processes mentioned in $C$; or results in $\none$, if there exists at least one process mentioned in $C$ on which $C$ is not projectable. Formally, the EPP of a choreography $C$ is defined in \cref{epp:chor:def:fig}.
\begin{figure}
\vspace{-1.0em}
\begin{subfigure}{1.0\textwidth}
\begin{align*}
    \epp{C} &\triangleq
    \begin{cases}
        \,\some{\pid{p}\mapsto P_{\pid{p}}} &\text{if for each $\pid{p} \in \pn(C)$ such that $\proj{C}{p} \eqr\some{P_{\pid{p}}}$, and $P_{\pid{p}} = \pnil$ when $\pid{p} \notin \pn(C)$}\\[-0.5ex]
        \,\none &\text{if there exists $\pid{p} \in \pn(C)$ such that $\proj{C}{p} \eqr \none$}
    \end{cases}
\end{align*}
\vspace{-1.0em}
\caption{Endpoint Projection for Choreographies}
\label{epp:chor:def:fig}
\end{subfigure}
\begin{subfigure}{1.0\textwidth}
\begin{align*}
    \epp{\mathscr{C}} &\triangleq
    \begin{cases}        
        \,\some{\{X_i(\pids{p}\setminus \pid{p_i}) = P_{\pid{p_i}}\}} &\text{if for each $X(\pids{p}) = C\in \mathscr{C},\, \proj{C}{p_i} \eqr \some{P_{\pid{p_i}}}$ and $\pid{p_i} \in \pids{p}$}\\[-0.5ex]
        \,\none &\text{if for any $X(\pids{p}) = C\in \mathscr{C}$, there exists a $\pid{p} \in \pids{p}$ such that $\proj{C}{p} \eqr \none$}
    \end{cases}
\end{align*}
\vspace{-1.0em}
\caption{Endpoint Projection for Choreographic Procedure Definitions}
\label{epp:c:def:fig}
\end{subfigure}
\label{epp:def:fig}
\vspace{-1.0em}
\caption{Definitions of Endpoint Projection}
\vspace{-1.0em}
\end{figure}
Utilising the distributivity property in proposition~\ref{epp:proc-projection:seqcomp}, we prove that EPP is also distributive over sequential composition shown in proposition~\ref{epp:epp:seqcomp}. These two properties are crucial for proving the correctness of EPP, especially for the conditional and choice cases as well as their out-of-order executions.
\begin{proposition}\label{epp:epp:seqcomp}
    For any choreography $C_1$ and $C_2$, 
    $\epp{C_1 \fatsemi C_2} = \epp{C_1} \fatsemi \epp{C_2}$.
\end{proposition}

Lastly, in \cref{epp:c:def:fig} we define the EPP for a set of procedure definitions $\mathscr{C}$, overloading the notation and written as $\epp{\mathscr{C}}$, which is a function of the type: $\epp{\cdot} : \mathfrak{C} \to \optg{\mathfrak{P}}$.
The EPP of $\mathscr{C}$ can either results in a set of process procedure definitions, if for each procedure definition $X(\pids{p}) = C$ in $\mathscr{C}$,
$C$ is projectable on $\pid{p}_i$ and results in $P_{\pid{p_i}}$,
where $\pid{p}_i$ is the $i$-th element of the parameter list $\pids{p}$, or it results in $\none$, if there exists at least one procedure definition $X(\pids{p}) = C$ in $\mathscr{C}$ and at least one process $\pid{p}$ in $\pids{p}$ such that $C$ is not projectable on $\pid{p}$. 
Thus, a choreographic procedure with $n$ process parameters generates $n$ process procedures, one for the perspective of each parameter.

\subsection{Coherence}
\label{epp:coherence}
\begin{figure}[b]
\vspace{-1.0em}
\begin{mathpar}
\inferrule*[right={\scriptsize coh-rt-call}]{\mathscr{C} \propto C \\ \mathscr{C} \propto C' \\ X(\pids{s}) = C'' \in \mathscr{C} \\ |\pids{p}| = |\pids{s}| \\ \proj{C''[\pids{p}/\pids{s}] \fatsemi C'}{q} \eqr\some{P} \\ \proj{C}{q} \eqr\some{Q} \\ \some{P} \mbrsw \some{Q}}
{\mathscr{C} \propto \pid{q} : X(\pids{p}).C' \seq C}
\end{mathpar}
\caption{Coherence of Choreographies (The Rule for Runtime Terms)}
\label{epp:coherence:fig-c}
\vspace{-1.0em}
\end{figure}
Coherence is a technical construction that ensures the correct projections of runtime terms used by recursive calls, providing consistency between the continuation of a runtime term and its procedure definition for a process that still have to enter.
The rule \textsc{coh-rt-call} in \cref{epp:coherence:fig-c} states that for $\pid{q} : X(\pids{p}).C' \seq C$ to be coherent w.r.t. a set of procedure definitions $\mathscr{C}$, the continuation $C$ of the runtime term $\pid{q} : X(\pids{p}).C'$, and the continuation $C'$ for $\pid{q}$, which is the process that still needs to enter the call, are both coherent w.r.t. $\mathscr{C}$. 
Specifically, $\proj{C''[\pids{p}/\pids{s}] \fatsemi C'}{q} \eqr\some{P}$ checks that projecting $C''[\pids{p}/\pids{s}] \fatsemi C'$, the unfolding of the procedure call in $X(\pids{p}); C'$ that has generated this runtime term, onto $\pid{q}$ yields a process term $P$, such that $\some{P} \mbrsw \some{Q}$, where $Q$ results from projecting the $C$ onto $\pid{q}$.
It means that $\pid{q}$ is not involved in more actions in $C$ than in $C''[\pids{p}/\pids{s}] \fatsemi C'$.
Although coherence mainly concerns runtime terms, we require all choreographies to be coherent w.r.t. $\mathscr{C}$, since the runtime terms can appear in any part of the choreography. 
All coherence rules other than \textsc{coh-rt-call} are straightforward (cf. \cref{app:epp:coherence}).
Similar to the preservation of well-formedness in theorem~\ref{chor:wf:preservation:thm}, coherence is also preserved by the semantics of choreographies shown in theorem~\ref{epp:coherence:preservation:thm}.
\begin{theorem}[Preservation of Coherence]\label{epp:coherence:preservation:thm}
    Given a well-formed choreographic configuration $\langle C, \cstore, \mathscr{C} \rangle$, both $C$ and $\mathscr{C}$ are projectable, and $C \propto \mathscr{C}$, if $\langle C, \cstore, \mathscr{C} \rangle \xrightarrow[]{\mu} \langle C', \cstore', \mathscr{C} \rangle$,
    then the resulting choreography $C'$ is still coherent w.r.t. $\mathscr{C}$, i.e., $C' \propto \mathscr{C}$.
\end{theorem}
\begin{proof}
        By structural induction on the derivation of the transition.
\end{proof}

\section{Correctness of EPP}
\label{sec:epp:correctness}
We first show the correctness of EPP for single-step transitions. Then we show the correctness of multi-step executions as a direct consequence of the correctness of EPP for single-step transitions.

\subsection{Completeness and Soundness of EPP}
\label{sec:epp:completeness:soundness}
The correctness of EPP for single-step transitions is established by proving the \emph{completeness} and \emph{soundness} of EPP. 
The completeness means that all networks resulting for EPP are \emph{complete} implementations of their corresponding choreographies. Specifically, it states that if a choreography $C$ can perform a transition to $C'$, then the network $N$ resulting from the EPP of $C$ can mimic the choreographic transition, and the resulting network has at least the same number of branches as the network resulting from the EPP of $C'$. Formally, the completeness theorem is shown in \cref{epp:completeness:thm}.
\begin{theorem}[Completeness of EPP]\label{epp:completeness:thm}
   Given a well-formed choreographic configuration $\langle C, \cstore, \mathscr{C} \rangle$ where $C \propto \mathscr{C}$, $\epp{C} \eqr \some{N}$, and $\epp{\mathscr{C}} \eqr\some{\mathscr{P}}$, for any transition label $\mu$, if $\langle C, \cstore, \mathscr{C} \rangle \xrightarrow[]{\mu} \langle C', \cstore', \mathscr{C} \rangle$, then there exists networks $N'$ and $M'$ such that $\langle N, \cstore, \mathscr{P} \rangle \xrightarrow[]{\mu} \langle N', \cstore', \mathscr{P} \rangle$, $\epp{C'} \eqr\some{M'}$, and $N' \mbrs M'$.
\end{theorem}
\begin{proof}
   By structural induction on the derivation of the transition, using
   theorems~\ref{net:semantics:transition-mergeable} and \ref{merge:net:transition} for relating out-of-order executions in choreographies and parallel executions in networks.
\end{proof}

Intuitively, the soundness is the converse of the completeness, meaning that all networks resulting from EPP are \emph{sound} implementations of their corresponding choreographies. It states that if a network $N$, which has at least many branches as the network resulting from the EPP of a choreography $C$ can perform a transition to $N'$, then such a transition is prescribed by the corresponding choreography $C$, and the resulting choreography can be projected to a network that does not have more branches than $N'$. Formally, the soundness theorem is shown in \cref{epp:soundness:thm}.
Note that we require $N \mbrs M$ for generalising the theorem to multi-step transitions discussed in \cref{sec:epp:multi-step}.
\begin{theorem}[Soundness of EPP]\label{epp:soundness:thm}
    Given a well-formed choreographic configuration $\langle C, \cstore, \mathscr{C} \rangle$ where $C \propto \mathscr{C}$, $\epp{C} \eqr \some{M}$, and $\epp{\mathscr{C}} \eqr\some{\mathscr{P}}$, for any networks $N$, $N'$, and transition label $\mu$ where $N \mbrs M$, if $\langle N, \cstore, \mathscr{P} \rangle \xrightarrow[]{\mu} \langle N', \cstore', \mathscr{P} \rangle$, then there exists a choreography $C'$ and a network $M'$ such that $\langle C, \cstore, \mathscr{C} \rangle \xrightarrow[]{\mu} \langle C', \cstore', \mathscr{C} \rangle$, $\epp{C'}\eqr\some{M'}$, and $N' \mbrs M'$.
\end{theorem}
\begin{proof}
    By induction on the choreography $C$ using \cref{net:semantics:transition-mergeable,merge:net:transition,chor:semantics:seqcomp-border}.
\end{proof}

\subsection{Correctness of Multi-Step Executions}
\label{sec:epp:multi-step}
We use a standard transitive closure to define the multi-step transitions for both choreographies and networks, written as $\langle C, \cstore, \mathscr{C} \rangle \xtwoheadrightarrow{\,\tls\,} \langle C', \cstore', \mathscr{C} \rangle$ and $\langle N, \cstore, \mathscr{P} \rangle \xtwoheadrightarrow{\,\tls\,} \langle N', \cstore', \mathscr{P} \rangle$.
The sequence $\tls$ is a list of transition labels that can be empty.
The detailed definitions are shown in \cref{app:multi}, and we implement the multi-step transition relation as \code{MTr} in CSLib.
We show the correctness of EPP for multi-step executions by the completeness \cref{epp:multi-step-completeness:thm} and soundness \cref{epp:multi-step-soundness:thm}:
\begin{theorem}[Completeness of Multi-Step Executions]\label{epp:multi-step-completeness:thm}
    Given a well-formed choreographic configuration $\langle C, \cstore, \mathscr{C} \rangle$ where $C \propto \mathscr{C}$, $\epp{C} \eqr \some{N}$, and $\epp{\mathscr{C}} \eqr\some{\mathscr{P}}$, for any list of transition label $\tls$, if $\langle C, \cstore, \mathscr{C} \rangle \xtwoheadrightarrow{\,\tls\,} \langle C', \cstore', \mathscr{C} \rangle$, then there exists networks $N'$ and $M'$ such that $\langle N, \cstore, \mathscr{P} \rangle \xtwoheadrightarrow{\,\tls\,} \langle N', \cstore', \mathscr{P} \rangle$, $\epp{C'} \eqr\some{M'}$, and $N' \mbrs M'$.
\end{theorem}

\begin{theorem}[Soundness of Multi-Step Executions]\label{epp:multi-step-soundness:thm}
    Given a well-formed choreographic configuration $\langle C, \cstore, \mathscr{C} \rangle$ where $C \propto \mathscr{C}$, $\epp{C} \eqr \some{M}$, and $\epp{\mathscr{C}} \eqr\some{\mathscr{P}}$, for any networks $N$, $N'$, and list of transition label $\tls$ where $N \mbrs M$, if $\langle N, \cstore, \mathscr{P} \rangle \xtwoheadrightarrow{\,\tls\,} \langle N', \cstore', \mathscr{P} \rangle$, then there exists a choreography $C'$ and a network $M'$ such that $\langle C, \cstore, \mathscr{C} \rangle \xtwoheadrightarrow{\,\tls\,} \langle C', \cstore', \mathscr{C} \rangle$, $\epp{C'} \eqr\some{M'}$, and $N' \mbrs M'$.
\end{theorem}
We prove both theorems by induction on the derivation of the multi-step transition. In addition to the completeness and soundness~\cref{epp:completeness:thm,epp:soundness:thm}, we also need the preservation of well-formedness (\cref{chor:wf:preservation:thm}), more branches relation (\cref{merge:mbrs:net:transition}), and coherence (\cref{epp:coherence:preservation:thm}).

\section{Consequences of the Correctness of Endpoint Projection}
\label{sec:epp:properties}
The correctness of EPP guarantees important safety and liveness properties of networks projected from choreographies. In this section, we show two properties, which are \emph{communication safety} in \cref{epp:comm-safety} and \emph{deadlock freedom} in \cref{epp:deadlock-freedom}, as consequences of the correctness of EPP.

\subsection{A Safety Property: Communication Safety}
\label{epp:comm-safety}
Communication safety ensures the absence of communication errors. Intuitively, these communication errors are caused by mismatched sending, receiving, selection, and branching actions.

To formally define a state to be communication safe, we first define \emph{communication actions} and \emph{duality} of communication actions in \cref{comm-act:fig}, as well as \emph{enabled communication actions} in a process term defined in \cref{enabled-comm-act:fig}. The communication actions is a set of actions including sending ($!$), receiving ($?$), selection ($\oplus_{\albl}$), and branching ($\&_{\albl}$). We use the notation $\overline{\alpha}$ to denote the dual of $ \alpha$. The dual of a $!$ is a $?$ and the dual of a $\oplus_{\albl}$ is a $\&_{\albl}$ with the same label $\albl$, and vice versa.
\begin{figure}
\vspace{-1.0em}
\begin{align*}
   \text{Communication Action (\setof{Act})}\quad \alpha \ \metaDeff \ ! \cmid \ ? \cmid \oplus_{\albl} \cmid \&_{\albl}
\end{align*}
\vspace{-1.5em}
\begin{align*}
\begin{aligned}[c]
\overline{!} \triangleq \ ?
\end{aligned}
\quad\quad
\begin{aligned}[c]
\overline{?} \triangleq \ !
\end{aligned}
\quad\quad
\begin{aligned}[c]
\overline{\oplus_{\albl}} \triangleq \&_{\albl}
\end{aligned}
\quad\quad
\begin{aligned}[c]
\overline{\&_{\albl}} \triangleq \oplus_{\albl}
\end{aligned}
\end{align*}
\vspace{-1.0em}
\caption{Communication Actions and Duality of Communication Actions}
\label{comm-act:fig}
\vspace{-1.0em}
\end{figure}

\begin{figure}[b]
\vspace{-1.0em}
\begin{align*}
\begin{aligned}[c]
    \pnil \octopusop \pid{p} &\triangleq \emptyset
\end{aligned}
\quad\quad
\begin{aligned}[c]
    x \metaDef e \seq P \octopusop \pid{p} &\triangleq \emptyset
\end{aligned}
\quad\quad
\begin{aligned}[c]
    X(\pids{q}) \seq P \octopusop \pid{p} &\triangleq \emptyset
\end{aligned}
\end{align*}
\vspace{-1.7em}
\begin{align*}
\begin{aligned}[c]
    \condif e\condthen P_1\condelse P_2 \seq P \octopusop \pid{p} &\triangleq \emptyset
\end{aligned}
\quad\quad
\begin{aligned}[c]
    P_1 \choice P_2 \seq P \octopusop \pid{p} &\triangleq P_1 \octopusop \pid{p} \cup P_2 \octopusop \pid{p}
\end{aligned}
\end{align*}
\vspace{-1.7em}
\begin{align*}
\begin{aligned}[c]
    \psend{q}{e}\seq P \octopusop \pid{p} &\triangleq
    \begin{cases}
        \ \{!\} &\text{if $\pid{p} = \pid{q}$}\\[-0.5ex]
        \ \emptyset &\text{otherwise}
    \end{cases}\\[-0.5ex]
    \psel{q}\seq P \octopusop \pid{p} &\triangleq
    \begin{cases}
        \ \{\oplus_{\albl}\} &\text{if $\pid{p} = \pid{q}$}\\[-0.5ex]
        \ \emptyset &\text{otherwise}
    \end{cases}
\end{aligned}
\quad\quad
\begin{aligned}[c]
    \precv{q}{x}\seq P \octopusop \pid{p} &\triangleq
    \begin{cases}
        \ \{?\} &\text{if $\pid{p} = \pid{q}$}\\[-0.5ex]
        \ \emptyset &\text{otherwise}
    \end{cases}\\[-0.5ex]
    \precvl{q} \octopusop \pid{p} &\triangleq
    \begin{cases}
        \ \{\&_{\albl_i}\}_{i\in I} &\text{if $\pid{p} = \pid{q}$}\\[-0.5ex]
        \ \emptyset &\text{otherwise}
    \end{cases}
\end{aligned}
\end{align*}
\caption{Enabled Communication Actions}
\label{enabled-comm-act:fig}
\vspace{-1.0em}
\end{figure}
If a communication action appears in the first instruction of a process term, then it is an enabled communication action.
We define $P \octopusop \pid{p}$ as the operation for gathering enabled communication actions in a process term.
For instance, if the first action of $P$ is a send instruction $\psend{q}{e}$ and $\pid{p}$ is the sender, then the enabled communication action for $\pid{p}$ in $P$ is $!$. 
If the first action of $P$ is a receive instruction $\precv{q}{x}$ and $\pid{p}$ is the receiver, then the enabled communication action for $\pid{p}$ in $P$ is $?$. 
If the first action of $P$ is a select instruction $\psel{q}$ and $\pid{p}$ is the sender, then the enabled communication action for $\pid{p}$ in $P$ is $\oplus_{\albl}$. 
If the first action of $P$ is a branching term $\precvlpr{q}{P}$ and $\pid{p}$ is the receiver, then the enabled communication actions for $\pid{p}$ in $P$ are $\&_{\albl_i}$ for each label $\albl_i$ in the branching term. In all other cases, there is no enabled communication action for $\pid{p}$, hence it returns an empty set.

We now give the formal definition of a communication safety state in definition~\ref{def:comm-safety-state}.
\begin{definition}[Communication Safe State]\label{def:comm-safety-state}
    A communication safe state for any network $N$ is a state where, for any two different processes $\pid{p}$ and $\pid{q}$,
    given $N (\pid{q}) \octopusop \pid{p} \neq \emptyset$,
    for any communication action $\alpha \in N (\pid{p}) \octopusop \pid{q}$, if $\alpha$ is a branching action, then there exists a branching action $\alpha'\in N (\pid{p}) \octopusop \pid{q}$ and $\overline{\alpha'} \in  N (\pid{q}) \octopusop \pid{p}$; 
    if $\alpha$ is not a branching action, then $\overline{\alpha} \in  N (\pid{q}) \octopusop \pid{p}$.
\end{definition}
In \cref{epp:comm-safety:thm}, we formulate and prove the communication safety property for networks resulting from EPP, which states that if a network $N$ resulting from EPP of a choreography $C$ can perform a multi-step transition to $N'$, then the resulting network $N'$ is still at a communication safe state.
\begin{theorem}[Communication Safety]\label{epp:comm-safety:thm}
    \hspace{-0.8ex}For any well-formed choreographic configuration $\langle C, \cstore, \mathscr{C} \rangle$ where $C \propto \mathscr{C}$, if $\epp{C} \eqr \some{N}$ and $\epp{\mathscr{C}} \eqr\some{\mathscr{P}}$, then for any network $N'$ and list of transition labels $\tls$ such that $\langle N, \cstore, \mathscr{P} \rangle \xtwoheadrightarrow{\,\tls\,} \langle N', \cstore', \mathscr{P} \rangle$, the resulting network $N'$ is at a communication safe state.
\end{theorem}
\begin{proof}
    We first show that given a choreography $C$ that is well-formed $\langle C, \cstore, \mathscr{C} \rangle \checkmark$ and $C \propto \mathscr{C}$, if $\epp{C} \eqr \some{N}$, the network $N$ is at a communication safe state by \cref{epp:comm-safe-state}.
    Then we prove the theorem using the soundness of EPP for multi-step executions stated in theorem~\ref{epp:multi-step-soundness:thm} together with the preservation of well-formedness in \cref{chor:wf:preservation:thm} and coherence in \cref{epp:coherence:preservation:thm}.
\end{proof}

\subsection{A Liveness Property: Deadlock Freedom}
\label{epp:deadlock-freedom}
Deadlock freedom is a property that ensures the absence of deadlocks, which are states where the system cannot make any progress. For choreographies and their projected networks, a deadlock can occur when there are processes waiting for each other to perform actions, resulting in a circular dependency that prevents any process from make a transition. As a consequence of the correctness of EPP, all networks resulting from EPP are deadlock free. To show the deadlock freedom of a network, we first introduce the \emph{progress property} of projectable choreographies in \cref{epp:progress:thm}. Then in \cref{epp:deadlock-freedom:thm}, we show the deadlock freedom of networks resulting from EPP.
\begin{theorem}[Progress of Projectable Choreographies]\label{epp:progress:thm}
    For any well-formed choreographic configuration $\langle C, \cstore, \mathscr{C} \rangle$ where $C \propto \mathscr{C}$, if both $C$ and $\mathscr{C}$ are projectable, and $C \neq \cnil$, then there exists a choreography $C'$, a store $\cstore'$, and a transition label $\mu$ such that $\langle C, \cstore, \mathscr{C} \rangle \xrightarrow[]{\mu} \langle C', \cstore', \mathscr{C} \rangle$.
\end{theorem}
\begin{proof}
    By structural induction on the derivation of the well-formedness of $C$.
\end{proof}
\begin{theorem}[Deadlock Freedom]\label{epp:deadlock-freedom:thm}
    For any well-formed choreographic configuration $\langle C, \cstore, \mathscr{C} \rangle$ where $C \propto \mathscr{C}$, if $\epp{C} \eqr \some{N}$ and $\epp{\mathscr{C}} \eqr\some{\mathscr{P}}$, and for any network $N'$, store $\cstore'$, and list of transition labels $\tls$ such that $\langle N, \cstore, \mathscr{P} \rangle \xtwoheadrightarrow{\,\tls\,} \langle N', \cstore', \mathscr{P} \rangle$ where $N' \neq \pnil$, then there exists a network $N''$, a store $\cstore''$, and a transition label $\mu$ such that $\langle N', \cstore', \mathscr{P} \rangle \xrightarrow[]{\mu} \langle N'', \cstore'', \mathscr{P} \rangle$.
\end{theorem}
\begin{proof}
    We prove the theorem by using the soundness of EPP for multi-step executions in \cref{epp:multi-step-soundness:thm} and the progress property of projectable choreographies in \cref{epp:progress:thm}, the completeness of EPP in \cref{epp:completeness:thm} together with the preservation of well-formedness and coherence.
\end{proof}

\section{Related Work}
\label{sec:related-work}

\subsubsection*{Mechanised theories of choreographic programming}

\begin{table}[t]
\centering
\caption{Comparison of mechanised choreographic programming languages.}
\label{tab:mechanised-cp}
\begingroup
\footnotesize
\addtolength{\tabcolsep}{-0.2em}
\begin{tabular}{@{}cccccccccc@{}}
\toprule
\multirow{2}{*}{Work} &
\multirow{2}{*}{Assistant} &
\multicolumn{2}{c}{Knowledge of Choice} &
\multicolumn{2}{c}{Nondeterminism} &
\multicolumn{3}{c}{Recursion} &
\multirow{2}{*}{Comm.} \\
\cmidrule(lr){3-4}
\cmidrule(lr){5-6}
\cmidrule(lr){7-9}
&
&
Select &
Merge &
Local &
Chor. &
Schema &
Parameters &
Entry & \\
\midrule
(1) & HOL4 & Binary  & Plain & No  & No  & Tail    & None         & Blocking   & Async \\
(2) & Rocq & Binary  & Full  & No  & No  & Tail    & None         & Concurrent & Sync  \\
(3) & Rocq & Binary  & Full  & Yes & No  & General & Higher-Order & Blocking   & Async \\
(4) & Rocq & Binary  & Full  & Yes & No  & General & Higher-Order, Process Sets & Blocking   & Async \\
(5) & Lean~4 & General & Full  & Yes & Yes & General & Processes    & Concurrent & Sync  \\
\midrule
\multicolumn{10}{c}{\footnotesize
(1) Kalas \cite{PGSN22}, 
(2) Core Choreographies \cite{CMP21b},
(3) Pirouette \cite{HG22},
}\\
\multicolumn{10}{c}{\footnotesize
(4) Quick Change \cite{SHC25},
(5) \mech (this work)
}\\
\bottomrule
\end{tabular}
\endgroup
\vspace{-1.0em}
\end{table}

Several mechanisations have been proposed in recent years. Although these developments share the goal of formalising a metatheory of CP, they differ primarily in three dimensions: their treatment of KoC, the forms of nondeterminism admitted by the language, and the expressiveness of recursion (cf.~\cref{tab:mechanised-cp}).

\emph{Knowledge of choice.~}
All these works address KoC during EPP through variations of the select-and-merge approach. 
We classify them according to whether selections express \emph{binary} or \emph{general} choices and whether merging is \emph{plain} or \emph{full}.\footnote{%
  Plain merge is the original merge operator introduced for multiparty session types~\cite{BCDLDY08} and it only merges branches at the top level.
  In contrast, full merge~\cite{DYBH12} recursively combines compatible continuations and can merge behaviours that differ deeper in the protocol structure reconciling more branches than plain merge.
}
Core Choreographies~\citep{CMP21b,CMP23} provides the first mechanisation of full merge. Restricting selections to two alternatives bounds the number of combinations that merge needs to handle to a fixed number, allowing these cases to be hard-coded into the definition. 
While this simplifies definitions and proofs, it reduces protocol efficiency, since multi-way choices are encoded as a cascade of binary selections. 
In contrast, \mech{} is the first mechanisation with general selections and full merge. 
Alternative approaches to KoC that are not mechanised are discussed later in this section. 

\emph{Nondeterminism.~}
All these works adopt a two-layer design in which choreographies are paired with a dedicated language for local computation. 
Pirouette~\cite{HG22} is the first mechanisation to support nondeterministic local evaluation. 
\mech{} is the only one to support computations whose nondeterminism arises from concurrency and communication, in addition to local computation, through the choreographic choice construct. 

\emph{Recursion.~}
We classify recursion according to the recursion schema, the type of parameters of recursive definitions, and whether processes must synchronise when entering a call. 
Pirouette is the first mechanised language to support general recursion and higher-order parameters, allowing choreographies of which participants are fixed to be passed as arguments.
However, Pirouette imposes a global synchronisation among \emph{all} processes whenever a choreographic function is invoked, which is a cost avoided by some models and implementations of Higher-Order CP~\cite{CGLMP23,GMP24}.
Quick Change~\cite{SHC25} extends Pirouette with abstractions over sets of processes~\cite{CM16,DY11}, enabling choreographies that are parametric over dynamic collections of participants with homogeneous behaviour. 
Core Choreographies is the first mechanisation to support recursion without introducing additional communications or synchronisations.

\emph{Communication.~}
Existing mechanisations span both synchronous and asynchronous (FIFO) semantics. 
This distinction is largely orthogonal to expressive power in the absence of choreographic nondeterminism~\cite{CM17b}.

\subsubsection*{Mechanised theories of multiparty session types}
Closely related to CP, \emph{multiparty session types} (MPST) study global specifications of distributed protocols together with their projection to local specifications \cite{HYC08,HYC16}. 
Unlike choreographies, the \emph{global types} in MPST describe communication structure rather than executable programs.
Consequently, mechanisations of MPST focus primarily on the correctness of typing and EPP \emph{without} computation.
Several mechanised developments of MPST have appeared in recent years
\cite{CGY21,JBK22,EY24,TBC25a,TBC25b,CFJ26}. %
With the exception of \cite{CFJ26}, which dispenses with EPP in favour of typing endpoint programs directly against global types, all mechanisations adopt for their EPP the standard MPST treatment of KoC through general selections and plain merge. %
All these works use closed tail-recursive definitions and none support choreographic nondeterminism.

\subsubsection*{Knowledge of choice}
Approaches to KoC proposed in the literature can be classified primarily on whether the propagation happens after the choice point, at the choice point itself, or before.

\emph{After choice.~}
The select-and-merge approach has underpinned CP since its introduction \cite{M13:phd,M23}. 
Classical selections carry only constant labels. 
While this is sufficient from a theoretical perspective, in practice, protocols often combine control-flow information with data in a single message to reduce communication overhead.
\citet{LM23} introduced \emph{type-based selection}, %
to convey KoC by the message type itself, and demonstrated interoperability with existing protocol implementations through the first choreographic implementation of IRC~\cite{K20}.

\emph{At choice.~}
A second line of work integrates KoC propagation directly into the semantics of conditionals.
Dynamic Choreographies~\cite{PGGLM16} and HasChor~\cite{SKK23} broadcast the outcome of each conditional to all its participants without relying on merge.
This substantially simplifies EPP but may introduce unnecessary communications, since processes with identical (or mergeable) behaviours across branches are nevertheless contacted. Also, it incurs the same inefficiency of binary select as nested conditionals require cascading broadcasts.

\emph{Before choice.~}
A third line of work establishes KoC before a choice point is reached, dispensing with selections altogether. 
Differently from the other approaches, it preserve KoC beyond the scope of individual conditionals.
\citet{JB22} formulated KoC as a logical predicate, using Hoare-style reasoning to prove that all participants in a choice possess sufficient information to reach the same outcome independently. 
\citet{BKJSKN25} introduced a type system that achieves KoC by ensuring that all participants in a choice evaluate the same deterministic guard on the same data. 
Both works incur the same communication cost as current `at-choice' approaches, as they require all participants to establish KoC even if their behaviours are not affected by the choice.
We anticipate that extending these works with merge can address this limitation.
\subsubsection*{Choreographic and Mixed choice}
In CP, \citet{M23} proposed the first choreographic choice primitive, sketching its semantics and EPP (without formal claims or proofs).
While the idea for EPP is reasonable, the semantics does not correctly capture compiled behaviour and does not guarantee deadlock-freedom. For instance, both examples~\ref{ex:chor-nested-choice} and \ref{ex:chor-gc} are counterexamples to the correctness of EPP under the sketched semantics in \citep{M23}, which we discovered thanks to our mechanisation efforts.
Our semantics for choreographic choice in \mech is new and resolves both issues.
To the best of our knowledge, no other CP language supports any form of mixed choice.

In MPST, more sophisticated and expressive forms of choice have been explored.
\citet{PY24a,PY24b} study mixed choice for synchronous MPST without requiring every choice to be resolved by a statically designated participant, as in \cite{M23} and \mech.
Their framework can represent communication patterns like the M and Star synchronisation~\cite{GGS08,PNG13}, which are inexpressible using conditionals or choreographic choices.
\citet{BHVT26} extended mixed choice to asynchronous MPST~\cite{BHVT26} allowing participants to speculatively execute different branches before a distinguished observer ultimately commits the protocol to one outcome. 
This permits a uniform treatment of phenomena such as exceptions, faults, and timeouts, but breaks the traditional behavioural equivalence between global and local specifications while the choice remains unresolved, highlighting the challenges of combining asynchronous communication with protocol-level nondeterminism.

\section{Conclusion and Future Work}
\label{sec:conclusion:future-work}
We have presented \mech, the first mechanisation of CP featuring general branching, general recursion, and choreographic nondeterminism. Our mechanisation effort is substantial, which contains more than 40,000 lines of Lean 4 code. It establishes safety and liveness guarantees for new scenarios that could not be correctly addressed in previous theories of CP such as first-come-first-served scheduling, locks, and barriers.

\mech is designed to enable future developments of CP based on our mechanisation.
We plan to extend \mech with asynchronous communication, as can improve efficiency \cite{CM13,PPM24} of communication. As recent work on asynchronous mixed choice for MPST~\cite{BHVT26} illustrates, the interaction between asynchrony nondeterministic behaviours introduces fundamental semantic challenges for establishing the equivalence between global and endpoint behaviours. On top of that, extending \mech with asynchrony also requires us to understand the interaction of these features with full merge.
There has been efforts to build certified compiler frameworks for CP~\cite{PGSN22,HG22,CMP21b,CLM23} using existing mechanised theories. With expressive features in \mech, we will also investigate how to build certified and semi-certified programming toolchains based on \mech that connect the benefits of our theory to mainstream tooling, for example by compiling our process terms to executable programming languages like Java or Rust. Here the formal parametrisation of our development explained in \cref{sec:parameters} -- for example on local computation -- is useful.


\bibliographystyle{ACM-Reference-Format}
\bibliography{references}

\appendix

\section{Stores and Process Name Substitutions}
\label{app:stores-subst}
\begin{definition}[Update Operations of Stores]\label{app:def:store-update}
\begin{align*}
\begin{aligned}[c]
\sigma[y \mapsto v] (x) &\triangleq
\begin{cases}
\ v &\text{if $x = y$}\\
\ \sigma(x) &\text{otherwise}
\end{cases}
\end{aligned}
\quad\quad\quad
\begin{aligned}[c]
\cstore[\pid{q}.x \mapsto v] (\pid{p}) &\triangleq
\begin{cases}
\ \cstore(\pid{p})[x \mapsto v] &\text{if $\pid{p} = \pid{q}$}\\
\ \cstore(\pid{p}) &\text{otherwise}
\end{cases}
\end{aligned}
\end{align*}
\end{definition}

\begin{definition}[Name Substitution for Process Names and Sequences of Process Names]\label{chor:semantics:pid-pids-substitution}
\begin{align*}
\begin{aligned}[c]
    \pid{p}[\pid{r}/\pid{s}] &\triangleq
    \begin{cases}
    \,\pid{r} &\text{if $\pid{p} = \pid{s}$}\\
    \,\pid{p} &\text{otherwise}
    \end{cases}
\end{aligned}
\quad\quad\quad\quad
\begin{aligned}[c]
    (\pid{p}_1, \dots \pid{p}_n) [\pid{r}/\pid{s}] &\triangleq (\pid{p}_1[\pid{r}/\pid{s}], \dots, \pid{p}_n[\pid{r}/\pid{s}])
\end{aligned}
\end{align*}
\end{definition}

\begin{definition}[Name Substitution for Choreographies and Instructions]\label{chor:semantics:chor-substitution}
\begin{align*}
\begin{aligned}[c]
\cnil[\pid{r}/\pid{s}] &\triangleq \cnil\\
(\pid{p}.x \metaDef e)[\pid{r}/\pid{s}] &\triangleq (\pid{p}[\pid{r}/\pid{s}]).x \metaDef e\\
(\gencom)[\pid{r}/\pid{s}] &\triangleq (\pid{p}[\pid{r}/\pid{s}].e) \tto (\pid{q}[\pid{r}/\pid{s}]).x\\
(\gensel)[\pid{r}/\pid{s}] &\triangleq (\pid{p}[\pid{r}/\pid{s}]) \tto (\pid{q}[\pid{r}/\pid{s}])[\albl]
\end{aligned}
\quad\quad
\begin{aligned}[c]
(I\seq C)[\pid{r}/\pid{s}] &\triangleq (I[\pid{r}/\pid{s}])\seq (C[\pid{r}/\pid{s}])\\
(X(\pids{p})[\pid{r}/\pid{s}]) &\triangleq X(\pids{p}[\pid{r}/\pid{s}])\\
(\pid{q} : X(\pids{p}).C) [\pid{r}/\pid{s}] &\triangleq \pid{q}[\pid{r}/\pid{s}] : X(\pids{p}[\pid{r}/\pid{s}]).(C[\pid{r}/\pid{s}])\\
(C_1 \choice_\pid{p} C_2) [\pid{r}/\pid{s}] &\triangleq (C_1[\pid{r}/\pid{s}]) \choice_{\pid{p}[\pid{r}/\pid{s}]}(C_2[\pid{r}/\pid{s}])
\end{aligned}
\end{align*}
\vspace{-1.0em}
\begin{align*}
(\condif \pid{p}.e\condthen C_1\condelse C_2) [\pid{r}/\pid{s}] &\triangleq \condif(\pid{p}[\pid{r}/\pid{s}]).e \condthen(C_1[\pid{r}/\pid{s}]) \condelse (C_2[\pid{r}/\pid{s}])\\
C[\pid{r}_1, \dots, \pid{r}_n/\pid{s}_1, \dots, \pid{s}_n] &\triangleq C[\pid{r}_1/\pid{s}_1]\dots[\pid{r}_n/\pid{s}_n]
\end{align*}
\end{definition}

\begin{definition}[Name Substitution for Process Terms and Instructions]\label{proc:semantics:proc-substitution}
\begin{align*}
\begin{aligned}[c]
\pnil[\pid{r}/\pid{s}] &\triangleq \pnil\\
(\precvl p)[\pid{r}/\pid{s}] &\triangleq \precvl p[\pid{r}/\pid{s}]\\
(\psend{p}{e})[\pid{r}/\pid{s}] &\triangleq \psend{p[\pid{r}/\pid{s}]}{e}\\
(x \metaDef e)[\pid{r}/\pid{s}] &\triangleq x \metaDef e
\end{aligned}
\quad\quad
\begin{aligned}[c]
(I\seq P)[\pid{r}/\pid{s}] &\triangleq (I[\pid{r}/\pid{s}])\seq (P[\pid{r}/\pid{s}])\\
(\psel{p})[\pid{r}/\pid{s}] &\triangleq \psel{p[\pid{r}/\pid{s}]}\\
(\precv{p}{x})[\pid{r}/\pid{s}] &\triangleq \precv{p[\pid{r}/\pid{s}]}{x}\\
(X(\pids{p})[\pid{r}/\pid{s}]) &\triangleq X(\pids{p}[\pid{r}/\pid{s}])
\end{aligned}
\end{align*}
\vspace{-1.0em}
\begin{align*}
(\condif e\condthen P_1\condelse P_2) [\pid{r}/\pid{s}] &\triangleq \condif e \condthen(P_1[\pid{r}/\pid{s}]) \condelse (P_2[\pid{r}/\pid{s}])\\
(P_1 \choice P_2) [\pid{r}/\pid{s}] &\triangleq (P _1[\pid{r}/\pid{s}]) \choice (P_2[\pid{r}/\pid{s}])\\
P[\pid{r}_1, \dots, \pid{r}_n/\pid{s}_1, \dots, \pid{s}_n] &\triangleq P[\pid{r}_1/\pid{s}_1]\dots[\pid{r}_n/\pid{s}_n]
\end{align*}
\end{definition}

\section{Well-formedness of Choreographies}
\label{app:chor:wf}
We define the \emph{well-formedness of choreographies} with the notation $\langle \mathscr{C}, \{\pids{r}\}, C \rangle \checkmark$.
It denotes that a choreography $C$ is well-formed under a set of choreographic procedures $\mathscr{C}$ and a set of processes that may be involved $\{\pids{r}\}$.
Figure~\ref{chor:wf:fig} gives the formal definition of well-formed choreographies. A terminated choreography is always well-formed, shown in rule \textsc{wf-end}.
A choreography starting with a local assignment is well-formed if the process $\pid{p}$ that performs the assignment is involved in the choreography and the continuation is well-formed, shown in the rule \textsc{wf-assign}.
Rules \textsc{wf-com} and \textsc{wf-sel} are similar. They state that a choreography starting with a communication or a selection is well-formed if the two processes communicating an expression or a label, namely, $\pid{p}$ and $\pid{q}$,
are different and both of them are involved in the choreography. In addition, the continuation of the choreography also needs to be well-formed.
Rules \textsc{wf-cond} and \textsc{wf-choice} are similar. For a choreography that starts with a conditional or a choice, if the guard process or the process that resolves the choice is involved in the choreography,
and both branches as well as the continuation are well-formed, then the choreography is well-formed.

The well-formedness of a choreography starting with a procedure call has more requirements, which is given in rule \textsc{wf-call}. It requires that all processes in the arguments of the procedure call $\pids{p}$ are distinct
and involved in the choreography, the invoked procedure $X(\pids{q}) = C'$ is defined, the number of arguments matches the number of parameters $\pids{q}$ of the invoked procedure, and the continuation is well-formed.
Lastly, rule \textsc{wf-rt-call} gives the well-formedness of a choreography starting with a runtime term. It states that if the arguments of the procedure call $\pids{p}$ are distinct,
the process $\pid{q}$ that still has to start the procedure call is involved in the choreography and the arguments of the procedure call, the invoked procedure $X(\pids{s}) = C''$ is defined. 
the number of arguments matches the number of parameters $\pids{s}$ of the invoked procedure, and the continuation is well-formed, then the choreography is well-formed.
\begin{figure}
\begin{mathpar}
\inferrule*[right={\scriptsize wf-end}]{ }
{ \langle \mathscr{C}, \{\pids{r}\}, \termination \rangle \checkmark}%

\inferrule*[right={\scriptsize wf-assign}]{\pid{p} \in \{\pids{r}\} \\ \langle \mathscr{C}, \{\pids{r}\}, C \rangle \checkmark} 
{ \langle \mathscr{C}, \{\pids{r}\}, \pid{p}.x \metaDef e \seq C \rangle \checkmark }%

\inferrule*[right={\scriptsize wf-com}]{\pid{p} \neq \pid{q} \\ \pid{p}, \pid{q} \in \{\pids{r}\} \\ \langle \mathscr{C}, \{\pids{r}\}, C \rangle \checkmark}  
{ \langle \mathscr{C}, \{\pids{r}\},\gencom \seq C \rangle \checkmark}%

\inferrule*[right={\scriptsize wf-sel}]{\pid{p} \neq \pid{q} \\ \pid{p}, \pid{q} \in \{\pids{r}\} \\ \langle \mathscr{C}, \{\pids{r}\}, C \rangle \checkmark }
{ \langle \mathscr{C}, \{\pids{r}\}, \gensel \seq C \rangle \checkmark }%

\inferrule*[right={\scriptsize wf-cond}]{ \pid{p} \in \{ \pids{r}\} \\ \langle \mathscr{C}, \{\pids{r}\}, C_1 \rangle \checkmark \\ \langle \mathscr{C}, \{\pids{r}\}, C_2 \rangle \checkmark \\ \langle \mathscr{C}, \{\pids{r}\}, C \rangle \checkmark}
{\langle \mathscr{C}, \{\pids{r}\}, \condif \pid{p}.e\condthen C_1\condelse C_2 \seq C \rangle \checkmark}%

\inferrule*[right={\scriptsize wf-choice}]{ \pid{p} \in \{\pids{r}\} \\ \langle \mathscr{C}, \{\pids{r}\}, C_1 \rangle \checkmark \\ \langle \mathscr{C}, \{\pids{r}\}, C_2 \rangle \checkmark \\ \langle \mathscr{C}, \{\pids{r}\}, C \rangle \checkmark}
{\langle \mathscr{C}, \{\pids{r}\}, C_1 \choice_\pid{p} C_2 \seq C \rangle \checkmark}

\inferrule*[right={\scriptsize wf-call}]{ \distinct(\pids{p}) \\ \{\pids{p}\} \subseteq \{\pids{r}\} \\ X(\pids{q}) = C' \in \mathscr{C} \\ |\pids{p}| = |\pids{q}| \\ \langle \mathscr{C}, \{\pids{r}\}, C \rangle \checkmark }
{\langle \mathscr{C}, \{\pids{r}\}, X(\pids{p}) \seq C  \rangle \checkmark}

\inferrule*[right={\scriptsize wf-rt-call}]{ \distinct(\pids{p}) \\ \pid{q} \in \{\pids{r}\} \\ \pid{q} \in \{\pids{p}\} \\ X(\pids{s}) = C'' \in \mathscr{C} \\ |\pids{p}| = |\pids{s}| \\ \langle \mathscr{C}, \{\pids{r}\}, C \rangle \checkmark}
{ \langle \mathscr{C}, \{\pids{r}\}, \pid{q} : X(\pids{p}).C' \seq C \rangle \checkmark }
\end{mathpar}
\caption{Well-Formedness of Choreographies}
\label{chor:wf:fig}
\end{figure}

\section{Algebraic Properties of Operators}
\label{app:algebraic-properties}
\begin{proposition} [Sequential Composition of Choreographies Forms a Monoid]\label{app:chor:seqcomp:prop}
\begin{align*}
    \cnil \fatsemi C &\triangleq C \quad\quad C \fatsemi \cnil \triangleq C &\text{(Identity element)}\\
    (C_1 \fatsemi C_2) \fatsemi C_3 &\triangleq C_1 \fatsemi (C_2 \fatsemi C_3) &\text{(Associativity)}
\end{align*}
\end{proposition}
\begin{proof}
By structure induction on the choreography.
\end{proof}

\begin{proposition} [Parallel Composition of Networks Forms a Monoid]\label{app:proc:parcomp:prop}
\begin{align*}
    \nnil \parcomp N &\triangleq N \quad\quad N \parcomp \nnil \triangleq N &\text{(Identity element)}\\
    N_1 \parcomp (N_2 \parcomp N_3) &\triangleq (N_1 \parcomp N_2) \parcomp N_3 &\text{(Associativity)}
\end{align*}
\end{proposition}
\begin{proof}
By unfolding the definition of parallel composition and function extensionality.
\end{proof}

\begin{proposition} [Sequential Composition of Process Terms Forms a Monoid]\label{app:proc:seqcomp:prop}
\begin{align*}
    \pnil \fatsemi P &\triangleq P \quad\quad P \fatsemi \pnil \triangleq P &\text{(Identity element)}\\
    (P_1 \fatsemi P_2) \fatsemi P_3 &\triangleq P_1 \fatsemi (P_2 \fatsemi P_3) &\text{(Associativity)}
\end{align*}
\end{proposition}
\begin{proof}
By structure induction on the process term.
\end{proof}

\begin{proposition} [Sequential Composition of Networks Forms a Monoid]\label{app:net:seqcomp:prop}
\begin{align*}
    \nnil \fatsemi N &\triangleq N \quad\quad N \fatsemi \nnil \triangleq N &\text{(Identity element)}\\
    (N_1 \fatsemi N_2) \fatsemi N_3 &\triangleq N_1 \fatsemi (N_2 \fatsemi N_3) &\text{(Associativity)}
\end{align*}
\end{proposition}
\begin{proof}
By pointwise lifting from proposition~\ref{app:proc:seqcomp:prop}.
\end{proof}

\begin{proposition} [Join-Semilattice Properties of $(\opt{Process}, \sqcup)$]\label{app:merge:semilattice:prop}
\begin{align*}
    \opterm{P} \sqcup \opterm{P} &\triangleq \opterm{P} \quad&\text{(Idempotency)}\\
    \opterm{P} \sqcup \opterm{Q} &\triangleq \opterm{Q} \sqcup \opterm{P} \quad&\text{(Commutativity)}\\
    (\opterm{P} \sqcup \opterm{Q}) \sqcup \opterm{R} &\triangleq \opterm{P} \sqcup (\opterm{Q} \sqcup \opterm{R}) \quad&\text{(Associativity)}
\end{align*}
\end{proposition}
\begin{proof}
By structure induction on the process term.
\end{proof}

\begin{proposition} [The Relation $\mbrs$ is a Partial Order]\label{app:merge:mbrs:prop}
\begin{align*}
    \opterm{P} \mbrs \opterm{P} &\quad \quad&\text{(Reflexivity)}\\
    \opterm{P} \mbrs \opterm{Q} \wedge \opterm{Q} \mbrs \opterm{P} &\Rightarrow  \opterm{P} = \opterm{Q} \quad&\text{(Antisymmetry)}\\
    \opterm{P} \mbrs \opterm{Q} \wedge \opterm{Q} \mbrs \opterm{R} &\Rightarrow  \opterm{P} \mbrs \opterm{R} \quad&\text{(Transitivity)}
\end{align*}
\end{proposition}
\begin{proof}
    By the definition of the relation $\mbrs$ and join-semilattice properties of $(\opt{Process}, \sqcup)$.
\end{proof}

\section{A Property of the Semantics of Choreographies}
\label{app:chor:lts:props}
If a choreography can perform a transition with label $\mu$ to another choreography, then choreography must contains all the processes mentioned in the transition label $\mu$, which is shown in theorem~\ref{chor:semantics:involved-processes}.
\begin{theorem}
\label{chor:semantics:involved-processes}
For any choreographic configuration $\langle C, \cstore, \mathscr{C} \rangle$ and transition label $\mu$, if $\langle C, \cstore, \mathscr{C} \rangle \xrightarrow[]{\mu} \langle C', \cstore', \mathscr{C} \rangle$, then \emph{$\pn(\mu) \subseteq \pn(C)$}.
\end{theorem}
\begin{proof}
    By structural induction on the derivation of the transition.
\end{proof}

\section{Properties of the Semantics of Networks}
\label{app:net:lts:props}
First of all, similar to the property of the semantics of choreographies described in theorem~\ref{chor:semantics:involved-processes}, we have the property that the processes involved in a network transition are exactly those that are mentioned in the transition label, shown in theorem~\ref{net:semantics:involved-processes}.
\begin{theorem}\label{net:semantics:involved-processes}
    For any network configuration $\langle N, \cstore, \mathscr{P} \rangle$ and transition label $\mu$, if $\langle N, \cstore, \mathscr{P} \rangle \xrightarrow[]{\mu} \langle N', \cstore', \mathscr{P} \rangle$, then \emph{$\pn(\mu)\subseteq \support(N)$}.
\end{theorem}
\begin{proof}
    By structural induction on the derivation of the transition.
\end{proof}
Based on this property, we establish two further properties of network transitions involving network restriction and process removal.
For network restriction, if a network can perform a transition, the its subnetwork containing only processed mentioned in the transition label can perform the same transition, shown in theorem~\ref{net:semantics:restriction}.
\begin{theorem}\label{net:semantics:restriction}
    For any network configuration $\langle N, \cstore, \mathscr{P} \rangle$ and transition label $\mu$, if $\langle N, \cstore, \mathscr{P} \rangle \xrightarrow[]{\mu} \langle N', \cstore', \mathscr{P} \rangle$, then \emph{$\langle N\restrict \pn(\mu), \cstore, \mathscr{P} \rangle \xrightarrow[]{\mu} \langle N'\restrict \pn(\mu), \cstore', \mathscr{P} \rangle$}.
\end{theorem}
\begin{proof}
    By structural induction on the derivation of the transition.
\end{proof}
Similarly, for process removal, if a network can perform a transition, then the resulting network after removing processes that are not mentioned in the transition label can also perform the same transition, shown in theorem~\ref{net:semantics:proc-removal}.
\begin{theorem}\label{net:semantics:proc-removal}
    For any network configuration $\langle N, \cstore, \mathscr{P} \rangle$ and transition label $\mu$, if $\langle N, \cstore, \mathscr{P} \rangle \xrightarrow[]{\mu} \langle N', \cstore', \mathscr{P} \rangle$, then for any finite set of processes  such that \emph{$\{\pids{p}\} \disjoint \pn(\mu)$}, $\langle N\setminus \{\pids{p}\} , \cstore, \mathscr{P} \rangle \xrightarrow[]{\mu} \langle N'\setminus \{\pids{p}\}, \cstore', \mathscr{P} \rangle$.
\end{theorem}
\begin{proof}
    By structural induction on the derivation of the transition.
\end{proof}
We also have an important properties stating that the mergeability of networks is preserved by network transitions. Specifically, if two networks $N$ and $M$ can both perform a transition with label $\mu$ to $N'$ and $M'$, respectively, and $N$ and $M$ are mergeable on all processes in $\mu$, then $N'$ and $M'$ remain mergeable.
\begin{theorem}[Preservation of Mergeability]\label{net:semantics:transition-mergeable}
    For any networks $N$, $M$, $N'$, and $M'$ such that $N \restrict {\emph{\pn}(\mu)} \sqcup M \restrict {\emph{\pn}(\mu)} \eqr \some{\hat{W}}$,
    if $\langle N, \cstore, \mathscr{P} \rangle \xrightarrow[]{\mu} \langle N', \cstore', \mathscr{P} \rangle$ and $\langle M, \cstore, \mathscr{P} \rangle \xrightarrow[]{\mu} \langle M', \cstore', \mathscr{P} \rangle$,
    then there exists finite networks $\hat{W'}$ such that $N'\restrict {\emph{\pn}(\mu)} \sqcup M'\restrict {\emph{\pn}(\mu)} \eqr \some{\hat{W'}}$.
\end{theorem}
\begin{proof}
    By induction on the transition label $\mu$.
\end{proof}

\section{Coherence}
\label{app:epp:coherence}
We present the formal definition of coherence between choreographies and their corresponding sets of choreographic procedures in figure~\ref{epp:coherence:fig}.
\begin{figure}
\begin{mathpar}
\inferrule*[right={\scriptsize coh-end}]{ }{\mathscr{C} \propto \cnil}

\inferrule*[right={\scriptsize coh-local}]{\mathscr{C} \propto C}{\mathscr{C} \propto \pid{p}.x \metaDef e \seq C}

\inferrule*[right={\scriptsize coh-com}]{\mathscr{C} \propto C}{\mathscr{C} \propto \gencom \seq C}

\inferrule*[right={\scriptsize coh-sel}]{\mathscr{C} \propto C}{\mathscr{C} \propto \gensel \seq C}

\inferrule*[right={\scriptsize coh-cond}]{\mathscr{C} \propto C_1 \\ \mathscr{C} \propto C_2\\ \mathscr{C} \propto C}{\mathscr{C} \propto \condif \pid{p}.e\condthen C_1\condelse C_2 \seq C}

\inferrule*[right={\scriptsize coh-choice}]{\mathscr{C} \propto C_1 \\ \mathscr{C} \propto C_2 \\ \mathscr{C} \propto C}{\mathscr{C} \propto  C_1 \choice_\pid{p} C_2 \seq C}

\inferrule*[right={\scriptsize coh-call}]{\mathscr{C} \propto C}{\mathscr{C} \propto X(\pids{p}) \seq C}

\inferrule*[right={\scriptsize coh-rt-call}]{\mathscr{C} \propto C \\ \mathscr{C} \propto C' \\ X(\pids{s}) = C'' \in \mathscr{C} \\ |\pids{p}| = |\pids{s}| \\ \proj{C''[\pids{p}/\pids{s}] \fatsemi C'}{q} \eqr\some{P} \\ \proj{C}{q} \eqr\some{Q} \\ \some{P} \mbrsw \some{Q}}
{\mathscr{C} \propto \pid{q} : X(\pids{p}).C' \seq C}
\end{mathpar}
\caption{Coherence of Choreographies}
\label{epp:coherence:fig}
\end{figure}

\section{Multi-step Transitions for Choreographies and Networks}
\label{app:multi}
The formal definitions of multi-step transitions for both choreographies and networks are shown in \cref{chor:semantics:multi-step} and \cref{net:semantics:multi-step}, respectively. 
\begin{figure}
\begin{subfigure}{1.0\textwidth}
\begin{mathpar}
    \inferrule*[right={\scriptsize refl}]{ }{\langle C, \cstore, \mathscr{C} \rangle \xtwoheadrightarrow{\,\epsilon\,} \langle C, \cstore, \mathscr{C} \rangle}

    \inferrule*[right={\scriptsize step-l}] {\langle C, \cstore, \mathscr{C} \rangle \xrightarrow[]{\mu} \langle C', \cstore', \mathscr{C} \rangle \\ \langle C', \cstore', \mathscr{C} \rangle\xtwoheadrightarrow{\,\tls\,}  \langle C'', \cstore'', \mathscr{C} \rangle }{\langle C, \cstore, \mathscr{C} \rangle \xtwoheadrightarrow{\,\mu \,\cons\, \tls\,} \langle C'', \cstore'', \mathscr{C} \rangle}
\end{mathpar}
\caption{Multi-Step Transitions for Choreographies}
\label{chor:semantics:multi-step}
\end{subfigure}
\begin{subfigure}{1.0\textwidth}
\begin{mathpar}
    \inferrule*[right={\scriptsize refl}]{ }{\langle N, \cstore, \mathscr{P} \rangle \xtwoheadrightarrow{\,\epsilon\,} \langle N, \cstore, \mathscr{P} \rangle}

    \inferrule*[right={\scriptsize step-l}] {\langle N, \cstore, \mathscr{P} \rangle \xrightarrow[]{\mu} \langle N', \cstore', \mathscr{P} \rangle \\ \langle N', \cstore', \mathscr{P} \rangle\xtwoheadrightarrow{\,\tls\,}  \langle N'', \cstore'', \mathscr{P} \rangle }{\langle N, \cstore, \mathscr{P} \rangle \xtwoheadrightarrow{\,\mu \,\cons\, \tls\,} \langle N'', \cstore'', \mathscr{P} \rangle}
\end{mathpar}
\caption{Multi-Step Transitions for Networks}
\label{net:semantics:multi-step}
\end{subfigure}
\caption{Multi-Step Transitions}
\label{app:semantics:multi-step}
\end{figure}
In both multi-step transitions, there are two rules. The rule \textsc{refl} is a reflexivity rule stating that any configuration can perform a multi-step transition with an empty transition label to itself. The rule \textsc{step-l} states that if a configuration can perform a single-step transition with a transition label $\mu$ to another configuration, and such a resulting configuration can perform a multi-step transition with a transition label $\tls$ to a third configuration, then the original configuration can perform a multi-step transition with a transition label obtained by concatenating $\mu$ and $\tls$ to the third configuration. In our Lean 4 mechanisation, we instantiate the multi-step transitions for choreographies and networks to the multi-step transition relation \code{MTr} in CSLib.

\section{An important property for proving communication safety}
As a consequence of EPP, all networks resulting from EPP are at communication safe states.
\begin{theorem}\label{epp:comm-safe-state}
    Given a choreography $C$ such that $\langle C, \pids{r}, \mathscr{C} \rangle \checkmark$ and $C \propto \mathscr{C}$, if $\epp{C} \eqr \some{N}$, the network $N$ is at a communication safe state.
\end{theorem}
\begin{proof}
    By structural induction on the choreography $C$.
\end{proof}

\end{document}